\documentclass[preprint,12pt,authoryear]{elsarticle}

\usepackage{xcolor}
\usepackage{psfrag}
\usepackage{balance}
\usepackage{mathrsfs}
\usepackage{dsfont}
\usepackage{comment}
\usepackage{graphicx}
\usepackage{amssymb,amsfonts}

\usepackage{amsthm}
\usepackage{mathalfa}
\usepackage{tikz}
\usetikzlibrary{calc}
\usepackage{color}
\usepackage{graphicx}
\usepackage{mathrsfs}
\usepackage{amsmath,amssymb,amsfonts,yhmath}
\usepackage{caption}
\usepackage{subcaption}
\usepackage{dsfont}

{\bf}{}
{\bf}{}
\newtheorem{theorem}{Theorem}
\newtheorem{prop}{Proposition}
\newtheorem{defi}{Definition}
\newtheorem{cor}{Corollary}
\newtheorem{example}{Example}{\bf}{}
{\bf}{}

\def\ones{\mathds{1}}

\newtheorem{remark}{Remark}

\journal{XYZ}

\begin{document}

\begin{frontmatter}



\title{Hybrid consensus for multi-agent systems with time-driven jumps}


\author{Andrea Cristofaro}
\author{Mattia Mattioni}
\address{{Dipartimento di Ingegneria Informatica, Automatica e Gestionale \emph{A. Ruberti} (Universit\`a degli Studi di Roma La Sapienza)},
          {via Ariosto 25}, 
           {00185} {Rome},
           {Italy}\\
           email: \texttt{andrea.cristofaro@uniroma1.it, mattia.mattioni@uniroma1.it} }

\begin{abstract}
In this paper, the behavior of scalar multi-agent systems over networks subject to time-driven jumps. Assuming that all agents communicate through distinct communication digraphs at jump and flow times, the asymptotic multi-consensus behavior of the hybrid network is explicitly characterized. The hybrid multi-consensus is shown to be associated with a suitable partition that is almost equitable for both the jump and flow communication digraphs. In doing so, no assumption on the underlying digraphs is introduced. Finally, the coupling rules making the multi-consensus subspace attractive are established. Several simulation examples illustrate the theoretical results.
\end{abstract}
\begin{keyword}
Multi-agent systems; 
Hybrid systems;
Consensus;
Impulsive systems;
Graph theory.
\end{keyword}
\end{frontmatter}

\section{Introduction}

Networked systems are nowadays well-considered a bridging paradigm among several disciplines spanning, among many others, from physics to engineering, psychology to medicine, biology to computer science. 
As typical in control theory \citep{isidori2017lectures}, we refer to a network (or multi-agent) system as composed of several dynamical units (agents) interconnected through a \emph{communication graph}: each node of the communication graph uniquely corresponds to one dynamical unit whereas edges model the corresponding exchange of information among agents. As a consequence, even for simple agents and with no issue in the network interconnection (e.g., time-delays), the network behavior is described by a complex dynamical system.

In this regard, several works have been devoted to providing methodological understanding on the collective behavior induced by the network interaction through the graph. Most of control problems involving network systems are generally related to driving all systems composing the network toward a \emph{consensus} behavior that might be global for all individuals or common only to some clusters uniquely identified by the graph (e.g., \citep{jadbabaie2003coordination, li2009consensus, aeyels2010emergence, ren2010distributed, chen2011cluster, egerstedt2012interacting, monshizadeh2015disturbance, wang2019collective}).
All of those works are devoted to the case of either continuous or discrete-time networks with possibly continuously time-varying communication topology. However, inspired by practical scenarios with heterogeneous samplings and discrete-event phenomena \citep{6338672, vetrella2017satellite, ferrante2015hybrid, de2018consensus, mattioni2020multiconsensus} and as deeply motivated in \cite{MAGHENEM2020335}, it is reasonable to allow for a hybrid communication topology, i.e., showing an interplay between a continuous-time behavior and event-triggered switchings. This setup embeds several practical scenarios, such as swarms of satellites communicating and sensing through different kind of sensors (i.e., radars and cameras), \citep{nag2013behaviour}, sampled-data agents at large \citep{barkai2021sampled} or cyber-security in networked systems \citep{duz2018stealthy}. To model the class of hybrid multi-agent systems considered in this paper, we adopt the general framework for analysis and design of hybrid control systems proposed in \citep{goebel2012hybrid}. Indeed, we consider a fixed group of autonomous agents whose coupling rules show two different regimes: a continuous-time flow of information associated with a given communication graph (the flow graph) and an intermittent update encoded by an other communication graph (the jump graph). In this set-up, we study the effects of the hybrid coupling functions and the network topology on the agents' trajectories. In particular, we are interested in the multi-consensus problem, that is, the definition of the coupling functions making the agents cluster into several subgroups possessing the same steady-state induced by the hybrid network. The single-consensus problem has been addressed in a similar hybrid setting \citep{7835664, ZHENG201951, guan2011consensus}. 
Assuming hence distinct (but general) digraphs characterizing the network during flow and jump times, the contribution of this paper is hence twofold: 
\begin{itemize}
    \item first, we characterize the number and type of consensuses over the hybrid multi-agent systems depending on the structure of the flow and jump digraphs with no assumption on the corresponding connectivity properties; 
    \item then, the coupling functions are designed in order to guarantee convergence to the consensus subspace.
\end{itemize}
In doing so, the consensus behaviors are explicitly characterized based on the notion of Almost Equitable Partitions (AEPs, \citep{godsil2013algebraic, cardoso2007laplacian}) firstly exploited for the characterization of multi-consensus in a purely continuous-time context in \citep{monaco2019multi}. In particular, we found that the multi-consensuses are uniquely identified by the joint effect of the flow and jump graphs and, in particular, the coarsest AEP that is common to both jump and flow communication digraphs; nodes belonging to the same cell of the partition converge to the same steady-state that is not necessarily a constant evolution but it is allowed to be a hybrid trajectory. 

\smallskip

The remainder of the paper is organized as follows. A review on graph theory and the multi-consensus problem for continuous-time networks is given in Section \ref{sec:recalls}. Then, the class of system under investigation is introduced in Section \ref{sec:hybrid_mas} and the problem formulated in Section \ref{sec:hybrid_cons_probl}. The main results are established in Section \ref{sec:main_result} with illustrating examples in Section \ref{sec:example}. Finally, conclusions are drawn in Section \ref{sec:conclusions}. 

\medskip
\emph{Notations.} $\mathbb{C}^+$ and $\mathbb{C}^-$ denote the right and left hand side of the complex plane. For a given a finite set $\mathcal{S}$, $|\mathcal{S}|$ denotes its cardinality. For a closed set $\mathcal{S}\subset \mathbb{R}^n$ and $x \in \mathbb{R}^n$, $\mathrm{dist}(x, \mathcal{S})$ denotes the distance of $x$ from the set $\mathcal{S}$. 
We denote by $0$ either the zero scalar or the zero matrix of suitable dimensions. $\ones_c$ denotes the $c$-dimensional column vector whose elements are all ones while $I$ is the identity matrix of suitable dimensions. Given a matrix $A \in \mathbb{R}^{n\times n}$ then $\sigma\{A \}\subset \mathbb{C}$ denotes its spectrum. A positive definite (semi-definite) matrix $A = A^\top$ is denoted by $A \succ 0$ ($A\succeq 0$); a negative definite (semi-definite) matrix $A = A^\top$ is denoted by $A\prec 0$ ($A\preceq 0$).

\section{Review on continuous-time systems over networks} \label{sec:recalls}
\subsection{Directed graph Laplacians and almost equitable partitions}

Let $\mathcal{G} = (\mathcal{V}, \mathcal{E})$ be a directed graph (or digraph for short) with $|\mathcal{V}| = N$, $\mathcal{E} \subseteq \mathcal{V} \times \mathcal{V}$. The set of neighbors to a node $\nu\in \mathcal{V}$ is defined as $\mathcal{N}(\nu) = \{ \mu \in \mathcal{V} \text{ s.t. } (\mu, \nu) \in \mathcal{E} \}$. For all pairs of distinct nodes $\nu, \mu \in \mathcal{V}$, a directed path from $\nu$ to $\mu$ is defined as $\nu \leadsto \mu := \{ (\nu_{r}, \nu_{r+1} )\in \mathcal{E} \text{ s.t. } \cup_{r = 0}^{{\ell-1}} (\nu_{r}, \nu_{r+1}) \subseteq \mathcal{E} \text{ with } \nu_{0} = \nu, \nu_{\ell} = \mu \text{ and } \ell >0\}$. The reachable set from a node $\nu \in \mathcal{V}$ is defined as $R(\nu) :=\{\nu\} \cup \{\mu\in \mathcal{V} \text{ s.t. }  \nu \leadsto \mu\} $. A set $\mathcal{R}$ is called a reach if it is a maximal reachable set, that is, $\mathcal{R} = R(\nu)$ for some $\nu\in\mathcal{V}$ and there is no $\mu\in \mathcal{V}$ such that $R(\nu) \subset R(\mu)$. Since $\mathcal{G}$ possesses a finite number of vertices, such maximal sets exist and are uniquely determined by the graph itself. Denoting by $\mathcal{R}_i$ for $i = 1, \dots, \mu$, the reaches of $\mathcal{G}$, the exclusive part of $\mathcal{R}_i$ is defined as $\mathcal{H}_i = \mathcal{R}_i/\cup_{\ell = 1, \ell \neq i}^\mu \mathcal{R}_\ell$ with cardinality $h_i = |\mathcal{H}_i|$. Finally, the common part of $\mathcal{G}$ is given by $\mathcal{C} = \mathcal{V}/\cup_{i = 1}^\mu \mathcal{H}_i$ with cardinality $c = |\mathcal{C}|$.

Given two graphs $\mathcal{G}_1 = (\mathcal{V}, \mathcal{E}_1)$ and $\mathcal{G}_2 = (\mathcal{V}, \mathcal{E}_2)$ one defines the \emph{union graph} $\mathcal{G} = \mathcal{G}_1 \cup \mathcal{G}_2 = (\mathcal{V}, \mathcal{E}_1 \cup \mathcal{E}_2)$ and, similarly, the \emph{intersection graph} as $\mathcal{G} = \mathcal{G}_1 \cap \mathcal{G}_2 = (\mathcal{V}, \mathcal{E}_1 \cap \mathcal{E}_2)$.

The Laplacian matrix associated with $\mathcal{G}$ is given by $L = D-A$ with $D\in \mathbb{R}^{N\times N}$ and $A\in\mathbb{R}^{N\times N}$ being respectively the in-degree and the adjacency matrices. As proved in \citep{agaev2005spectra}, $L$ possesses one eigenvalue $\lambda = 0$ with multiplicity coinciding with $\mu$, the number of reaches of $\mathcal{G}$, and the remaining $N-\mu$ with positive real part. Hence, after a suitable re-labeling of nodes, the Laplacian always admits the lower triangular form \citep{caughman2006kernels}
\begin{align}\label{lap_gen}
 L = \begin{pmatrix}
 L_1 & \dots & 0 & 0\\
 & \ddots & & \vdots \\
 0 & \dots & L_\mu & 0\\
 M_1 & \dots & M_\mu & M
 \end{pmatrix}
\end{align}
where: $L_i \in \mathbb{R}^{h_i \times h_i}$ ($i = 1, \dots, \mu$) is the Laplacian associated with the subgraph $\mathcal{H}_i$ and possessing one eigenvalue in zero with single multiplicity; $M\in \mathbb{R}^{c \times c}$ verifying $\sigma(M)\subset \mathbb{C}^+$ corresponds to the common component $\mathcal{C}$. Thus, the eigenspace associated with $\lambda = 0$ for $L$ is spanned by the right eigenvectors 
\begin{align}\label{eig:L}
z_1 = \begin{pmatrix}
\ones_{h_1} \\ 
\vdots \\
0\\
\gamma^1
\end{pmatrix} \quad \dots \quad z_\mu = \begin{pmatrix}
0 \\ 
\vdots \\
\ones_{h_\mu}\\
\gamma^\mu
\end{pmatrix}
\end{align}
with $\sum_{i = 1}^{\mu} \gamma^i = \ones_c$ and $M_i \ones_{h_i} + M \gamma^i  = 0$ for all $i = 1, \dots, \mu$. In addition, the left eigenvectors associated with the zero eigenvalues are given by
\begin{align}\label{left_eig}
    \tilde v_1^\top = \begin{pmatrix} v_1^\top & \dots &  0 & 0
    \end{pmatrix}\  \dots \  \tilde v_\mu^\top = \begin{pmatrix} 0 & \dots & v_\mu^\top & 0
    \end{pmatrix}
\end{align}
with $v_i^\top = \begin{pmatrix} v_i^1 & \dots & v_i^{h_i} \end{pmatrix} \in \mathbb{R}^{1 \times h_i}$, $v_i^s>0$  if the corresponding node is root or zero otherwise.
A partition $\pi = \{\rho_1, \dots, \rho_r\}$ of $\mathcal{V}$ is a collection of cells $\rho_i\subseteq \mathcal{V}$ verifying $\rho_i \cap \rho_j = \emptyset$ for all $i \neq j$ and $\cup_{i = 1}^r \rho_i = \mathcal{V}$. The characteristic vector of $\rho\subseteq \mathcal{V}$ is given by $p(\rho) = (p_1(\rho)\  \dots p_N(\rho))^\top \in \mathbb{R}^N$ with for $i = 1, \dots, N$
\begin{align*}
    p_i(\rho) = \begin{cases}
    1 & \text{if }  v_i \in \rho
    \\ 0 & \text{otherwise}.
    \end{cases}
\end{align*}
For a partition $ \pi = \{\rho_1, \dots, \rho_r\}$ of $\mathcal{V}$, the characteristic matrix of $\pi$ is $P(\pi) = \begin{pmatrix} p(\rho_1) & \dots & p(\rho_r)  \end{pmatrix}$ with $\mathcal{P} = \text{Im} P(\pi)$ with, by definition of partition, each row of $P(\pi)$ possessing only one element equal to one and all other being zero. Given two partitions $\pi_1$ and $\pi_2$, $\pi_1$ is said to be finer than $\pi_2$ ($\pi_1\preceq \pi_2$) if all cells of $\pi_1$ are a subset of some cell of $\pi_2$ so implying $\text{Im}P(\pi_2)\subseteq \text{Im}P(\pi_1)$; equivalently, we say that $\pi_2$ is coarser than $\pi_1$ ($\pi_2 \succeq \pi_1$), with $\text{Im}P(\pi_1)\subseteq \text{Im}P(\pi_2)$. We name $\pi = \mathcal{V}$ the trivial partition as composed of a unique cell with all nodes.
Given a cell $\rho\in \mathcal{V}$ and a node $\nu_i \notin \rho$, we denote by $\mathcal{N}(\nu_i,\rho) = \{ \nu \in \rho \text{ s.t } (\nu, \nu_i) \in \mathcal{E} \}$  the set of neighbors of $\nu_i$ in the cell $\rho$. 
\smallskip
\begin{defi}\label{def:AEP}
 A partition $\pi^\star = \{ \rho_1, \rho_2, \ldots, \rho_k \}$ is said to be
 an \emph{almost equitable partition (AEP)} of $\mathcal{G}$ if, for each $i,j \in \{1, 2, \ldots, k \}$, with $i \neq j$, there exists an integer $d_{ij}$ such that $| \mathcal{N}(\nu_i,\rho_j) | = d_{ij}$ for all $\nu_i \in \rho_i$.
\end{defi}
\smallskip
In other words, a partition such that each node in $\rho_i$ has the same number of
neighbors in $\rho_j$, for all $i,j$ with $i \neq j$, is an AEP.
The property of almost equitability is equivalent to the invariance of the subspaces generated by the characteristic vectors of its cells. In particular, we can give the following equivalent
characterization of an AEP  $\pi^\star$ \citep{monaco2019multi,monshizadeh2015disturbance}.
\smallskip
\begin{prop}\label{def:AEP_bis}
 Consider a graph $\mathcal{G}$ and a partition $\pi^\star = \{ \rho_1, \rho_2, \ldots, \rho_k \}$ with $\mathcal{P}^\star  = \text{Im} P(\pi^\star )$. $\pi^\star$ is an \emph{almost equitable partition (AEP)} if and only if $\mathcal{P}^\star$ is $\mathcal{L}$-invariant, that is,
 \begin{equation}\label{eq:prop_AEP}
 L\mathcal{P}^\star \subseteq \mathcal{P}^\star .
\end{equation}
\end{prop}
We say that a non trivial partition $\pi^\star$ is the coarsest AEP of $\mathcal{G}$ if  $\pi^\star \succeq \pi$ for all non trivial $\pi$ AEP of $\mathcal{G}$ and, equivalently, $\text{Im}P(\pi^\star) \subseteq \text{Im}P(\pi)$.  Algorithms for computing almost equitable partitions are available for arbitrary unweighted digraphs such as, among several others, the one in \citep{monshizadeh2015disturbance}.

\subsection{Multiconsensus of continuous-time systems}

As proved in \citep{monaco2019multi}, when considering multi-agent systems the notion of AEP is linked to the characterization of \emph{multi-consensus}. Roughly speaking, consider a set of $N>1$ of scalar integrators of the form
\begin{align*}
    \dot x_i = u_i
\end{align*}
with $x_i \in \mathbb{R}$ and continuously exchanging information based on a communication graph $\mathcal{G} = (\mathcal{V}, \mathcal{E})$ whose vertices $\nu_i \in \mathcal{V}$ correspond to the $i^\text{th}$ agent with state $x_i$ ($i = 1, \dots, N$). Then, under the coupling rule
\begin{align*}
    u_i = -\sum_{\ell{ : } \nu_\ell \in \mathcal{N}(\nu_i)} (x_i - x_\ell)
\end{align*}  
 nodes asymptotically cluster into $r = \mu+k$ consensuses for some $k\in \mathbb{N}$ such that $1\leq k \leq c$; those clusters are uniquely defined by the coarsest AEP $\pi^\star = \{\rho_1, \dots, \rho_r\}$ of $\mathcal{G}$: the states of all agents belonging to the same cell of the AEP converge to the same consensus state. More in detail, with a slight abuse of notation and exploiting (\ref{lap_gen}),the AEP underlying the multi-consensus of the network is $\pi^\star = \{\mathcal{H}_1, \dots, \mathcal{H}_\mu, \rho_{\mu+1}, \dots, \rho_{\mu+k} \}$ with  $\mathcal{C} = \cup_{\ell = 1}^k \rho_{\mu+\ell} $ and all nodes in $\rho_{\mu+\ell}$ (of cardinality $c_\ell  = |\rho_{\mu+\ell} |$) sharing the same components of the vectors $\gamma^i$ for all $i = 1, \dots, \mu$. More precisely, setting in (\ref{eig:L})
 \begin{align*}
     \gamma^i = \begin{pmatrix}
     \gamma^i_1 \\
     \vdots \\
     \gamma^i_{c}
     \end{pmatrix}\in\mathbb{R}^c, \quad i = 1, \dots, \mu
 \end{align*}
a node $\nu_{N-c+m_1}\in \mathcal{C}$ (with  $m_1 \in \{1, \dots, c\}$) belongs to the cell $\rho_{\mu+\ell}\subseteq \mathcal{C}$ if and only if $  \gamma^i_{m_1} = \gamma^i_{m_2}$ for all $\nu_{N-c+m_2}\in \rho_{\mu+\ell}$, $i = 1, \dots, \mu$ and some  $m_2 \in \{1, \dots, c\}$. Accordingly, the number of cells partitioning $\mathcal{C}$ (i.e., the integer $k$) coincides with the number of distinct coefficients of the vector $\gamma^i$ with $i = 1, \dots, \mu$.

Consider now the agglomerate network dynamics
\begin{align*}
    \dot x = -L x
\end{align*}
with $x = (x_1 \ \dots x_N)$, $x_i \in \mathbb{R}$, $L$ the communication graph Laplacian, and the sub-states  when $c_0 = h_0 = 0$
\begin{align*}
    \textbf{x}_i =& \text{col}\{x_{h_1+\dots +h_{i-1} +1}, \dots, x_{h_1+\dots +h_i} \} \in \mathbb{R}^{h_i} \\
    \textbf{x}_\ell =&
     \text{col}\{x_{N-c+c_1+\dots + c_{\ell-1} }, \dots, x_{N-c+c_1+\dots + c_{\ell}} \} \in \mathbb{R}^{c_\ell}
\end{align*} 
for  $i = 1, \dots, \mu$ and $\ell = 1, \dots, k$. One can further rewrite each $\gamma^i \in \mathbb{R}^c$ in (\ref{eig:L}) as
\begin{align}\label{eq:gammai_n}
\gamma^i = \begin{pmatrix}
\gamma^i_1 \ones_{c_1} \\
\vdots \\
\gamma^i_{k}\ones_{c_k}
\end{pmatrix}\in\mathbb{R}^c, \quad \gamma^i_\ell \in \mathbb{R}, \quad \ell = 1, \dots, k
\end{align}
with $c_\ell$ the cardinality of the cell $\rho_{\mu+\ell} \subseteq \mathcal{C}$ (i.e., $c_{\ell} = |\rho_{\mu+\ell}|$) so that, as $t\to \infty$, the following holds true: 
\begin{enumerate}
    \item nodes in the same reach $\mathcal{H}_i$ converge to the same consensus value; namely, for $i = 1, \dots, \mu$ and all $\nu_j \in \mathcal{H}_i$
    \begin{align*}
        \mathbf{x}_i(t) \to x_{i}^{ss} \ones_{h_i}, \quad x_{i}^{ss}:= v_i^\top \mathbf{x}_i(0)\in \mathbb{R}
    \end{align*}
    with $v_i^\top\in \mathbb{R}^{h_i}$ being the corresponding component of the left eigenvector (\ref{left_eig}); 
    \item nodes in $\mathcal{C}$ belonging to the same cell $\rho_{\mu+\ell}$ converge to a convex combination of the consensuses induced by the reaches; namely,
    \begin{align*}
        \mathbf{x}_\ell(t) \to x^{ss}_{c_\ell}\ones_{c_\ell},\quad x^{ss}_{c_\ell} := \sum_{i = 1}^\mu \gamma^{i}_\ell x_{i}^{ss}
    \end{align*}
    and $\gamma^i_\ell \in \mathbb{R}$ the distinct components of the vector $\gamma^i\in \mathbb{R}^c$ in (\ref{eq:gammai_n}), for $\ell = 1, \dots, k$.
\end{enumerate}

\section{Hybrid multi-agent systems} \label{sec:hybrid_mas}
We consider a group of $N\in\mathbb{N}$ identical agents, whose state is assumed to be a scalar $x_i\in\mathbb{R}$ for any $i=1,...,N$.
The evolution of each agent state is assumed to be governed by a hybrid dynamics, i.e, characterized by the interplay of a continuous-time and a discrete-time behaviour \citep{goebel2012hybrid}. The alternate selection of continuous and discrete dynamics can be either driven by specific time patterns or triggered by conditions on the state. In this paper we consider agents whose dynamics is given by a hybrid integrator with time-driven jumps. In particular, the state of each agent is assumed to obey the following update law:
\begin{equation}\label{eq:hybrid_agent}
\begin{array}{ll}
\dot{x}_i(t)=u_i(t)& t\in\mathbb{R^+}\setminus\mathcal{J}\smallskip\\
x_i^+(t)=x_i(t)+v_i(t)& t\in\mathcal{J}
\end{array}\quad i=1,...,N
\end{equation}
 where $\mathcal{J}$ denotes the sequence of jump instants
 $$
 \mathcal{J}=\{t_j \in\mathbb{R}^+, j=1,...,\aleph_\mathcal{J}: t_{j}<t_{j+1},\ \aleph_{\mathcal{J}}\in\mathbb{N}\cup\infty\}
 $$
 and where $u_i(t),v_i(t)\in\mathbb{R}$ are control inputs and $$\tau_\text{min}<t_{j+1}-t_j<\tau_\text{max}\qquad \forall j\in\mathcal{J}.$$ The differential equation in (\ref{eq:hybrid_agent}) is referred to as {\it flow dynamics}, whereas the difference equation corresponds to the {\it jump dynamics}. To keep track of the jumps, it is convenient to introduce the notion of hybrid time domain as a special case of \citep[Definition~2.3]{goebel2012hybrid}.\smallskip
 \begin{defi}
 A hybrid time domain is a set $\mathcal{T}$ in $[0,\infty)\times\mathbb{N}$ defined as the union of {\it indexed intervals}
\begin{equation}\label{eq:time_domain}
\mathcal{T}:=\bigcup_{j\in\mathbb{N}}\big\{\left[t_j,t_{j+1}\right]\times \{j\}\big\}
\end{equation}
Given an hybrid time domain, its length is defined as $length(\mathcal{T})=\sup_t\mathcal{T}+\sup_j\mathcal{T}.$
A hybrid time domain is said  $\tau$-periodic, for some constant $\tau>0$, if $t_{j+1} - t_{j} = \tau$ for any $j \in \mathbb{N}. \hfill\triangle$  
 \end{defi}\smallskip
 In view of the latter definition, we can enhance the notation for the state of the agents as follows
 $$
 x_i(t,j)\ \textrm{with\ }(t,j)\in\mathcal{T},\ i=1,...,N.
 $$
 As previously mentioned, the agents are supposed to be connected through a suitable communication graph. However, due to the hybrid nature of the problem, it is reasonable to allow interactions among the agents with different topologies for the flow dynamics and the jump dynamics. To this end, let us consider a flow graph $\mathcal{G}_\mathrm{f}=(\mathcal{V}_\mathrm{f},\mathcal{E}_\mathrm{f})$ and a jump graph $\mathcal{G}_\mathrm{j}=(\mathcal{V}_\mathrm{j},\mathcal{E}_\mathrm{j})$ where the sets of vertices satisfy
 $$
\mathcal{V} := \mathcal{V}_\mathrm{f}=\mathcal{V}_\mathrm{j}
 $$
 as the number of agents remains constant upon jumps. Accordingly, we aim at designing the control inputs in a decentralized way as functions of the states of agent $i$ and its neighbours, i.e.,
 $$
 \begin{array}{l}
 u_i=f_i(x_i,\{x_k:\nu_k \in\mathcal{N}_{\mathrm{f},i}\})\smallskip\\
 v_i=g_i(x_i,\{x_k:\nu_k\in\mathcal{N}_{\mathrm{j},i}\})
 \end{array}
 $$
where the set of indexes
$$
\mathcal{N}_{\mathrm{f},i}=\{\nu_k\in \mathcal{V}:(\nu_k, \nu_i)\in\mathcal{E}_\mathrm{f}\},\quad \mathcal{N}_{\mathrm{j},i}=\{\nu_k\in \mathcal{V}:(\nu_k,\nu_i)\in\mathcal{E}_\mathrm{j}\}
$$
are, respectively, the flow and the jump neighbourhoods of the agent $i$.

\begin{remark}       
Although impulsive multi-agent systems might be modelled in other possible ways (e.g., \citep{haddad2014impulsive, liberzon2003switching}), the choice of recasting the problem in the hybrid formalism is motivated by two main reasons: $(i)$ it is more compact and easily extendable to the case of general hybrid systems (with state-driven jumps too); $(ii)$ it allows to easily handle the geometric properties independently of the jump time-instants and flow intervals. 
\end{remark}

\section{The hybrid consensus problem} \label{sec:hybrid_cons_probl}
It is well known that, in classical multi-agent systems, consensus conditions can be attained by implementing simple linear feedback laws, which can be encoded through the Laplacian matrix of the communication graph. We propose a similar approach here, by setting
\begin{equation}\label{eq:inputs}
\begin{array}{l}
u_i(t)=-\sum_{k : \nu_k \in\mathcal{N}_{\mathrm{f},i}}(x_i(t)-x_k(t))\smallskip \\
v_i(t)=-\alpha\sum_{k: \nu_k \in\mathcal{N}_{\mathrm{j},i}}(x_i(t)-x_k(t))
\end{array}
\end{equation}
where $\alpha>0$ is a suitable gain to be tuned. Defining the cumulative state $x=[x_1\ x_2\ \cdots\ x_N]^\top $, the dynamics of the multi-agent system driven by (\ref{eq:inputs}) can be written as
\begin{equation}\label{eq:hybrid_sys}
\begin{array}{ll}
\dot x(t)=-L_\mathrm{f}x(t)& t\in \mathbb{R}^+\setminus\mathcal{J}\smallskip\\
x^+(t)=(I-\alpha L_{\mathrm{j}})x(t)&t\in\mathcal{J}
\end{array}
\end{equation}
where $L_\mathrm{f}$ and $L_\mathrm{j}$ denote, respectively, the Laplacian matrix of the flow and the jump graph. It is worth noticing that, whilst the negative flow Laplacian $-L_\mathrm{f}$ provides a marginally stable continuous dynamics, the jump map has to be made stable by means of a proper tuning of $\alpha>0$. To this end, we can invoke the following technical result.
\begin{prop}\label{prop:spectrum}
{\it Let $\mathcal{G}=(\mathcal{V},\mathcal{D})$ be a graph with $\bar{N}$ vertices, and let $L\in\mathbb{R}^{\bar{N} \times \bar{N}}$ be its Laplacian matrix. There exists $\alpha>0$ such that the spectrum of $(I-\alpha L)$ satisfy
\begin{eqnarray}
&\sigma(I-\alpha L)\subset \{z\in\mathbb{C}:|z|<1\}\cup\{1\}\label{eq:spectrum1}\smallskip\\
&1\in\sigma(I-\alpha L)\label{eq:spectrum2}
\end{eqnarray}}
\end{prop}
\begin{proof}
By construction \citep[Theorem 4.6]{veerman2020primer} we know that $\sigma(L)\subset\mathbb{C}^+$ with $\text{dim}\{\ker L\} = \mu \geq 1$, i.e, there exists $w\in\mathbb{R}^{\bar N}$ with
$Lw=0.$ Accordingly, for any $\alpha>0$, one has $(I-\alpha L)w=Iw-\alpha Lw=w$, this showing that $1\in\sigma(I-\alpha L)$. Let us now define
\begin{equation}\label{eq:alphastar}
    \alpha^*:=\min_{\lambda\in\sigma(L)\setminus\{0\}}\frac{2\mathrm{Re}(\lambda)}{|\lambda|^2}.
\end{equation}
Then selecting $\alpha\in(0,\alpha^*)$ guarantees the fulfilment  of condition (\ref{eq:spectrum1}). To verify this fact, we can observe that for any $\lambda\in\sigma(L)\setminus\{0\}$ the number $\tilde{\lambda}=1-\alpha\lambda\in\mathbb{C}$ is a eigenvalue for $I-\alpha L$, with
$$
|\tilde{\lambda}|=\sqrt{(1-\alpha\mathrm{Re}(\lambda))^2+\mathrm{Im}(\lambda)^2}.
$$
By simple algebraic manipulations one can see that $|\tilde{\lambda}|<1$ whenever~$\alpha<\alpha^\star$ with $\alpha^\star$ given by \eqref{eq:alphastar}.
\end{proof}
Our goal is to establish conditions on the combination of graphs $\mathcal{G}_\mathrm{f},\mathcal{G}_\mathrm{j}$ under which a consensus is achieved in the multi-agent system (\ref{eq:hybrid_sys}), and to characterize such consensus.

\section{Main result} \label{sec:main_result}
In this section, the main results are given to characterize the consensus behaviors of the hybrid network and the corresponding clusters.

\subsection{The multi-consensus clusters}

In the following result, the multi-consensus clusters arising from (\ref{eq:hybrid_sys}) under the hybrid connection in \eqref{eq:inputs} are characterized based on a suitably defined AEP for both the components of the topology. 

\begin{theorem}
Consider the hybrid multi-agent system \eqref{eq:hybrid_sys} and let $\mathcal{G}_\mathrm{f}$ and $\mathcal{G}_\mathrm{j}$ be the flow and jump graphs with Laplacians $L_\mathrm{f}$ and $L_\mathrm{j}$. Then, the multi-consensuses of \eqref{eq:hybrid_sys} are induced by the coarsest almost equitable partition $\pi_\mathcal{H}^\star$ for $\mathcal{G}_\mathrm{f}$ and $\mathcal{G}_\mathrm{j}$ simultaneously. 
\end{theorem}

\smallskip
\begin{proof}
For proving the result, first one needs to show that if nodes in the same cell $\rho_i \in \pi_\mathcal{H}^\star$ are initialized with the same initial condition, then, the corresponding evolutions are identical for all time. To this end, denoting $\mathcal{P}^\star_\mathcal{H} = \text{Im}P(\pi^\star_\mathcal{H})$, one gets that by definition $L_\mathrm{f} \mathcal{P}_\mathcal{H}^\star \subset \mathcal{P}_\mathcal{H}^\star$ and $L_\mathrm{j} \mathcal{P}_\mathcal{H}^\star \subset \mathcal{P}_\mathcal{H}^\star$. 
Picking $x_0 = x(t_0, 0) \in \mathcal{P}_\mathcal{H}^\star$ (i.e., agents in the same cell of $\pi_\mathcal{H}^\star$ share the same initial condition\footnote{By definition of partition, all rows of $P(\pi^\star_\mathcal{H})$ possess only one element equal to $1$ and all others are $0$. The generic vector $x = (x_1, \dots, x_N) \in \mathcal{P}_\mathcal{H}^\star $ is such that $x_\theta = x_\eta$ if and only if $\nu_\theta, \nu_\eta \in \mathcal{V}$ (with $\theta\neq \eta$ and $\theta, \eta = 1, \dots, N$) belong to the same cell of $\pi^\star_{\mathcal{H}}$.}), the evolution of (\ref{eq:hybrid_sys}) is given by
\begin{align*}
    x(t, j) = e^{-(t-t_j)L_\mathrm{f}}E(t_{j}, t_{j-1}) \dots E(t_1, t_0) x_0
\end{align*}
with $E(t_{j}, t_{j-1}) = (I-\alpha L_j)e^{-(t_{j}-t_{j-1})L_\mathrm{f}}$ the monodromy matrix.
By the induction principle, assume that $x(t_{j}, {j}) \in \mathcal{P}_\mathcal{H}^\star$. Then, as $\mathcal{P}_\mathcal{H}^\star$ is $L_\mathrm{f}$-invariant, as $t \in [t_{j+1}, t_{j})$, one gets
$ x(t, j) = e^{-(t-t_j)L_\mathrm{f}} x(t_{j}, {j}) \in \mathcal{P}_\mathcal{H}^\star$ and, in particular, $x(t_{j+1}^-, j) = e^{-(t_{j+1}-t_j)L_\mathrm{f}} x(t_{j}, {j}) \in \mathcal{P}_\mathcal{H}^\star$ as $t \to t_{j+1}^-$; namely, when flowing, trajectories of nodes belonging to the same cell of $\pi_\mathcal{H}^\star$ remain identical. Also, as $t = t_{j+1}$, one concludes  $x(t_{j+1}, j+1) = (I-\alpha L_\mathrm{j}) x(t_{j+1}^-, j) \in \mathcal{P}_\mathcal{H}^\star$ by $L_\mathrm{j}$-invariance of $\mathcal{P}_\mathcal{H}^\star$ (Proposition \ref{def:AEP_bis}). Accordingly, as $x_0 = x(t_0, 0) \in \mathcal{P}_\mathcal{H}^\star$, one gets that nodes of (\ref{eq:hybrid_sys}) in the same cell in $\pi_{\mathcal{H}}^\star$ share the same initial conditions for all $(t, t_j) \in \mathcal{T}$ and thus consensus. Finally, because $\pi_{\mathcal{H}}^\star$ is the coarsest AEP for both $\mathcal{G}_\mathrm{f}$ and $\mathcal{G}_\mathrm{j}$, then $\mathcal{P}^\star_\mathcal{H}$ is the largest invariant subspace (spanned by a partition) that is shared, at the same time, by $L_\mathrm{f}$ and $L_\mathrm{j}$ so getting the result.
\end{proof}

\smallskip

\begin{remark}
Denote by $\pi^\star_\mathrm{f}$ and $\pi^\star_\mathrm{j}$ the coarsest AEPs for, respectively, $\mathcal{G}_\mathrm{f}$ and $\mathcal{G}_\mathrm{j}$ verifying
\begin{align*}
    L_q \mathcal{P}^\star_{q} \subseteq \mathcal{P}^\star_{q}
\end{align*}
with $\mathcal{P}^\star_q = \text{Im}P(\pi^\star_q) \equiv \ker{L_q}$  being the (invariant) eigenspace associated with the non-zero eigenvalues of $L_q$ for $q \in \{ \mathrm{f}, \mathrm{j} \}$.
Because $\pi_\mathcal{H}^\star$ is the coarsest AEP for both $\mathcal{G}_\mathrm{f}$ and $\mathcal{G}_\mathrm{j}$, then the subspace 
\begin{align}\label{eq:calP}
    \mathcal{P}_\mathcal{H}^\star = \text{Im}P(\pi^\star_\mathcal{H})
\end{align}is $L_q$-invariant for $q \in \{\mathrm{f}, \mathrm{j} \}$. As a consequence, one necessarily gets $\mathcal{P}_\mathcal{H}^\star \subseteq \mathcal{P}_\mathrm{f}^\star$ and $\mathcal{P}_\mathcal{H}^\star \subseteq  \mathcal{P}_\mathrm{j}^\star$ that implies $\pi^\star_\mathrm{f}\preceq \pi^\star_\mathcal{H}$ and $\pi^\star_\mathrm{j}\preceq \pi^\star_\mathcal{H}$. 
\end{remark}

\begin{remark} 
In the proof of the result above, we have intentionally defined the consensus of (\ref{eq:hybrid_sys}) as a \emph{shared evolution} rather than a set of common asymptotic values the nodes tend to. 
\end{remark}

At this point, we show that $\pi_\mathcal{H}^\star$ can be constructed starting from the union and the intersection graphs, respectively associated with the flow and jump graphs. To this end, denoting by $\mathcal{G}_\mathrm{un} = \mathcal{G}_\mathrm{f} \cup \mathcal{G}_\mathrm{j}$ and $\mathcal{G}_\text{int} = \mathcal{G}_\mathrm{f} \cap \mathcal{G}_\mathrm{j}$, it is always possible \citep{monshizadeh2015disturbance} to sort nodes of the hybrid network in such a way the Laplacian $L_\mathrm{un}$ associated with $\mathcal{G}_\mathrm{un}$ gets the form
\begin{align}\label{Lu}
    L_\mathrm{un} = \begin{pmatrix}
    L_{\mathrm{un} ,1} & \dots & 0 & 0\\
    \vdots & \ddots & \vdots & \vdots \\
    0 &\dots & L_{\mathrm{un} ,\mu} & 0\\
    M_{\mathrm{un} ,1} & \dots & M_{\mathrm{un} , \mu} & M_\mathrm{un} 
    \end{pmatrix}
\end{align}
where $L_{\mathrm{un} ,i}$ are the Laplacians associated with the exclusive reaches $\mathcal{H}_{\mathrm{un},i}$ of $\mathcal{G}_\mathrm{un}$ and $M_\mathrm{un}$ the nonsingular matrix associated with the common $\mathcal{C}_\mathrm{un}$ of $\mathcal{G}_\mathrm{un}$. As a consequence, exploiting the relation \citep{jadbabaie2003coordination}
\begin{align}\label{Lu:eq}
L_\mathrm{un} = L_\mathrm{f} + L_\mathrm{j}-L_\text{int}
\end{align}
one gets for $q \in \{\text{int}, \mathrm{j}, f\}$
\begin{align}\label{Lq}
    L_q = \begin{pmatrix}
    L_{q,1} & \dots & 0 & 0\\
    \vdots & \ddots & \vdots & \vdots \\
    0 &\dots & L_{q,\mu} & 0\\
    M_{q,1} & \dots & M_{q, \mu} & M_q
    \end{pmatrix}.
\end{align}
Starting from this, the following result can be proved. 
\begin{theorem}\label{th:AEP_c}
Consider the hybrid multi-agent system \eqref{eq:hybrid_sys} and let $\mathcal{G}_\mathrm{f}$ and $\mathcal{G}_\mathrm{j}$ be the flow and jump graphs with Laplacians $L_\mathrm{f}$ and $L_\mathrm{j}$. Let  $\mathcal{H}_{\mathrm{un},1}, \dots, \mathcal{H}_{\mathrm{un},\mu}$ and $\mathcal{C}_\mathrm{un}$ be respectively the exclusive reaches and the common part of the union graph $\mathcal{G}_\mathrm{un} = \mathcal{G}_\mathrm{f}\cup \mathcal{G}_\mathrm{j}$ with $\mathcal{H}_{\mathrm{un},i}\cap \mathcal{C}_\mathrm{un} = \emptyset$ and $\mathcal{H}_{\mathrm{un},i}\cap \mathcal{H}_{\mathrm{un},j} = \emptyset$ for all $i\neq j\in \{1, \dots, \mu \}$. Let  $\pi_\mathcal{C} = \{ \rho_{\mu+1}, \dots, \rho_{\mu+k}\}$ be the coarsest partition of $\mathcal{C}_\mathrm{un}$ with $ \mathcal{C}_\mathrm{un} =  \cup_{\ell = 1}^k \rho_{\mu+\ell}$,  $c_\ell =|\rho_{\mu + \ell}| $ for $\ell = 1, \dots, k$ and characteristic matrix
\begin{align*}
   & \begin{pmatrix}
    p(\rho_{\mu+1}) & \dots & p(\rho_{\mu+k})
    \end{pmatrix} = \begin{pmatrix}
    0 & \dots & 0\\
     p_c(\rho_{\mu+1}) & \dots & p_c(\rho_{\mu+k})
    \end{pmatrix} \\
    & \mathcal{P}_c = \text{Im}P_c, \quad P_c = \begin{pmatrix} p_c(\rho_{\mu+1}) & \dots & p_c(\rho_{\mu+k}) \end{pmatrix}
\end{align*}
verifying
$M_{\mathrm{int} } \mathcal{P}_c \subseteq \mathcal{P}_c$ with $M_\mathrm{int} \in \mathbb{R}^{c \times c}$ as in (\ref{Lq}).
Then, the coarsest almost equitable partition for $\mathcal{G}_\mathrm{f}$ and $\mathcal{G}_\mathrm{j}$ is provided by 
\begin{align}   \label{eq:AEP_star}
    \pi_{\mathcal{H}}^\star = \pi_{\mathrm{un}}^\star:= \{\mathcal{H}_{\mathrm{un},1}, \dots, \mathcal{H}_{\mathrm{un},\mu}, \rho_{\mu+1}, \dots, \rho_{\mu+k}\}.
\end{align}
\end{theorem}
\begin{proof}
For proving the result one first needs to show that: $(i)$ $\pi_{\mathrm{un}}^\star$ is an AEP for both $\mathcal{G}_\mathrm{f}$ and $\mathcal{G}_\mathrm{j}$; $(ii)$ it is the coarsest one (i.e., $\pi_{\mathcal{H}}^\star = \pi_{\mathrm{un}}^\star$).

\noindent As far as $(i)$ is concerned, it suffices to show that $L_q \mathcal{P}_\mathrm{un}^\star \subseteq \mathcal{P}_\mathrm{un}^\star$ for $q \in \{\mathrm{f},\mathrm{j} \}$ and
\begin{align*}
    \mathcal{P}_\mathrm{un}^\star := \text{Im}P(\pi^\star_\mathrm{un}).
\end{align*} 
To this end, we first note that, by definition of partition, all cells share no node; namely, for all $i\neq j \in \{ 1, \dots, \mu\}$ 
\begin{align*} 
&\mathcal{H}_{\mathrm{un},i} \cap \mathcal{H}_{\mathrm{un},j} = \emptyset
\\
&\mathcal{C}_\mathrm{un}  =\cup_{\ell = 1}^k \rho_{\mu+\ell} \\
&\mathcal{H}_{\mathrm{un},i} \cap \mathcal{C}_\mathrm{un} = \emptyset.
\end{align*}
Because for $i = 1, \dots, \mu$, the cells associated with the exclusive reaches are given by $\mathcal{H}_{\mathrm{un} ,i} = \{\nu_{h_1+\dots+h_{i-1}+1}, \dots,  \nu_{h_1+\dots+h_{i}}\}$ with cardinality $h_i = |\mathcal{H}_{\mathrm{un},i} |$ and $h_0 = 0$, the corresponding characteristic vectors are given by
\begin{align*}
    p(\mathcal{H}_{\mathrm{un},1}) = \begin{pmatrix}
        \ones_{h_1} \\
        \vdots\\
        0\\
        0
    \end{pmatrix}, \dots , p(\mathcal{H}_{\mathrm{un},\mu}) = \begin{pmatrix}
        0\\
        \vdots\\
        \ones_{h_\mu}\\
        0
    \end{pmatrix}.
\end{align*}
Consequently, the characteristic matrix of $\pi^\star_\mathrm{un}$ gets the form
\begin{align*}
    P(\pi^\star_\mathrm{un}) =& \begin{pmatrix}
    \ones_{h_1} &\dots & 0 & 0\\
    \vdots & \ddots & \vdots & \vdots \\
    0 & \dots & \ones_{h_\mu} & 0 \\
    0 & \dots & 0 & P_c
     \end{pmatrix}.
\end{align*}
At this point, exploiting the structure of $L_q$ and the definition of $L_{q,i}$ in (\ref{Lq}), one easily gets $$L_{q}  p(\mathcal{H}_{\mathrm{un},i}) = \begin{pmatrix}
    0 \\ 
    \vdots \\
   L_{q,i} \ones_{h_i}\\
   \vdots \\
   0 \\ 0
\end{pmatrix} = 0$$ for $q \in \{\mathrm{f}, \mathrm{j}\}$ and $i = 1, \dots, \mu$ because $\ones_{h_i} \in \text{ker}{L_{q,i}}$. Accordingly, one needs to prove that $M_q \mathcal{P}_c \subseteq \mathcal{P}_c$ for $q \in \{\mathrm{f}, \mathrm{j}\}$. For this purpose, with no loss of generality, assume nodes in $\mathcal{C}$ are sorted in such a way that $\rho_{\mu+\ell} =  \{ v_{\mu + c_{1} + \dots +c_{\ell-1}+1}, \dots, v_{\mu + c_{1} + \dots +c_{\ell}}\}$ for $\ell = 1, \dots, k$ so that one rewrites
\begin{align}\label{eq:pc}
    p_c(\rho_{\mu+1}) = \begin{pmatrix}
    \ones_{c_1} \\
    \vdots \\
    0
    \end{pmatrix} \quad \dots  \quad     p_c(\rho_{\mu+k}) = \begin{pmatrix}
    0 \\
    \vdots \\
    \ones_{c_k}
    \end{pmatrix}.
\end{align}
Also, we rewrite $M_q \in \mathbb{R}^{c\times c}$ in (\ref{Lq}) as 
\begin{align}\label{eq:Mqr}
    M_{q} = \begin{pmatrix}
    M_{q,11} & \dots & M_{q,1k}\\
    \vdots & \ddots & \dots \\
    M_{q,k1} & \dots & M_{q,kk}
    \end{pmatrix}, \quad q \in \{\text{int}, \mathrm{un}, \mathrm{f},\mathrm{j}\}
\end{align}
with each $M_{q,\ell r} \in \mathbb{R}^{c_\ell \times c_r}$ being the component of the Laplacian establishing the in-coming connections among the cells $\rho_{\mu+\ell}$ and $\rho_{\mu+r}$ with $r, \ell = 1, \dots, k$ in the graph $\mathcal{G}_q$.
Accordingly, each element of the column vector $-M_q p_c(\rho_{\mu+\ell})$ represents the amount of neighbors that the corresponding node in ${\rho}_{\mu+\ell}\subseteq \mathcal{C}_\mathrm{un}$ possesses in the cell ${\rho}_{\mu+r}\subseteq \mathcal{C}_\mathrm{un}$ with $r, \ell = 1, \dots, k$ and $r \neq \ell$.
Noticing that the generic vector in $\mathcal{P}_c$ is given by
\begin{align*}
    \begin{pmatrix}
        \lambda_1 \ones_{c_1}\\
        \vdots \\
        \lambda_k \ones_{c_k}
    \end{pmatrix}
\end{align*}
for some real constants $\lambda_1, \dots, \lambda_k \in \mathbb{R}$, the relation $M_q \mathcal{P}_c \subseteq \mathcal{P}_c$ holds if and only if for all $ p_c(\rho_{\mu+\ell})$ and $\ell = 1, \dots, k$, one gets
\begin{align*}
    M_q p_c(\rho_{\mu+\ell}) = \begin{pmatrix}
    \varepsilon_{q,\ell}^1 \ones_{c_1} \\
    \vdots\\
    \varepsilon_{q,\ell}^k \ones_{c_k}
    \end{pmatrix}
\end{align*}
for some $\varepsilon^r_{q, \ell} \in \mathbb{Z}$ with $q \in \{\mathrm{f},\mathrm{j}\}$ and $r= 1, \dots, k$ \citep[Theorem~1]{monaco2019multi}. In particular, considering the graph $\mathcal{G}_q$ one gets that: for $r \neq \ell$, $-\varepsilon^r_{q, \ell}\geq 0$ represents the number of neighbours each node in $\rho_{\mu+\ell}$ possesses in $\rho_{\mu+r}$ under the graph $\mathcal{G}_q$; for $r = \ell$, $\varepsilon^\ell_{q, \ell}\geq 0$ is the number of neighbors that each node in  $\rho_{\mu+\ell}$ possesses in all other cells. By assumption and Definition \ref{def:AEP}, for $q = \text{int}$ (i.e., when considering the intersection graph $\mathcal{G}_\text{int}$) one gets
\begin{align*}
    M_\text{int} p_c(\rho_{\mu+\ell}) = \begin{pmatrix}
    \varepsilon_{\text{int},\ell}^1 \ones_{c_1} \\
    \vdots\\
    \varepsilon_{\text{int},\ell}^k \ones_{c_k}
    \end{pmatrix}
\end{align*}
with: $-\varepsilon_{\text{int},\ell}^r$ the number of neighbours each node in $\rho_{\mu+\ell}$ possesses in $\rho_{\mu+r}$ for $r\neq \ell$;  $\varepsilon_{\text{int},\ell}^\ell$ the number of neighbours that each node in $\rho_{\mu+\ell}$ possesses in the union of all other cells of the partition. 
Accordingly, by  Definition \ref{def:AEP} of AEP and Proposition \ref{def:AEP_bis}, one must show that the number of such neighbors is the same in both the jump and flow graphs  $\mathcal{G}_\mathrm{j}$ and $\mathcal{G}_\mathrm{f}$, that is for $q \in \{\mathrm{f}, \mathrm{j} \}$ and $\ell = 1, \dots, k$
\begin{align*}
    M_q p_c(\rho_{\mu+\ell}) = \begin{pmatrix}
    \varepsilon_{q,\ell}^1 \ones_{c_1} \\
    \vdots\\
    \varepsilon_{q,\ell}^k \ones_{c_k}
    \end{pmatrix}, \quad  \varepsilon_{q,\ell}^k\in \mathbb{Z}.
\end{align*}
To this end, exploiting (\ref{Lu:eq}), one gets the equality 
\begin{align}\label{eq:Min_s}
    M_\text{int} p_c(\rho_{\mu+\ell}) = (M_\mathrm{f} + M_\mathrm{j} -M_\mathrm{un})p_c(\rho_{\mu+\ell}).
\end{align}
By the structure (\ref{eq:pc}) and (\ref{eq:Mqr}), equality (\ref{eq:Min_s}) rewrites blockwise for $r = 1, \dots, k$ as
\begin{align*}
    M_{\mathrm{f}, r\ell}\ones_{c_\ell} +  M_{\mathrm{j}, r\ell}\ones_{c_\ell} -  M_{\mathrm{un} , r\ell}\ones_{c_\ell} = \varepsilon^r_{\text{int},\ell} \ones_{c_r}.
\end{align*}
By definition of $\varepsilon_{\text{int},\ell}^r$ , one has that for $q \in \{\mathrm{un}, \mathrm{f},\mathrm{j} \}$
\begin{align*}
M_{q, r\ell}\ones_{c_\ell} = \varepsilon_{\text{int},\ell}^r\ones_{c_\ell} + \alpha_{q,r\ell}
\end{align*}
for some vector $\alpha_{q, r\ell} \in \mathbb{R}^{c_\ell}$. The components of the vector $\alpha_{q, r\ell}$ coincide with: the cumulative number of exclusive (that is, when excluding the shared ones in the intersection graph) neighbors the corresponding node possesses in $\rho_{\mu+r}$ if $r\neq \ell$; the opposite of the total number of not shared (again, when excluding the ones in the intersection graph) neighbors that such a node possesses in all other cells otherwise.
Based on this, the equality above reduces to
$
\alpha_{\mathrm{un},r\ell}  = \alpha_{\mathrm{f},r\ell} + \alpha_{\mathrm{j},r\ell}
$
implying that, for all $r, \ell = 1, \dots, k$, necessarily $\alpha_{q, r\ell} = \hat \varepsilon_{q,r\ell} \ones_{c_r} $ for some constant $\hat\varepsilon_{q,r\ell}\in \mathbb{Z}$ and $q \in \{ \mathrm{f}, \mathrm{j}\}$. The first part of the proof is then concluded.

\smallskip
$(ii)$ Exploiting (\ref{Lu}) and the rewriting (\ref{Lq}), it turns out that necessarily $\{\mathcal{H}_{\mathrm{un},1}, \dots,\mathcal{H}_{\mathrm{un},\mu} \}$ is the coarsest non-trivial AEP for both $\mathcal{G}_\mathrm{j}$ and $\mathcal{G}_\mathrm{f}$ when considering only the exclusive reaches $\mathcal{H}_\mathrm{un}  = \mathcal{H}_{\mathrm{un},1}\cup \dots \cup \mathcal{H}_{\mathrm{un},\mu}$. For the remaining part of the proof, let us proceed by contradiction and assume that there exists an AEP $\bar \pi = \{\mathcal{H}_{\mathrm{un},1}, \dots,\mathcal{H}_{\mathrm{un},\mu}, \bar \rho_{\mu+1}, \dots, \bar \rho_{\mu+\bar k} \}$ that is coarser than $\pi^\star_\mathrm{un}$ (i.e., $\bar \pi \succ \pi_\mathrm{un}^\star$). Equivalently, the corresponding subspace $\bar{\mathcal{P}}_c = \text{Im}\bar P_c$, where $\bar P_c = \begin{pmatrix}
p(\bar \rho_{\mu+1}) & \dots & p(\bar \rho_{\mu+k})
\end{pmatrix}$, is invariant for both $M_\mathrm{j}$ and $M_\mathrm{f}$ but not $M_\text{int}$-invariant. As in the previous case, let the nodes be sorted in such a way that 
\begin{align*}
    p_c(\bar \rho_{\mu+1}) = \begin{pmatrix}
    \ones_{\bar c_1} \\
    \vdots \\
    0
    \end{pmatrix} \quad \dots  \quad     \bar p_c(\rho_{\mu+\bar k}) = \begin{pmatrix}
    0 \\
    \vdots \\
    \ones_{\bar c_{\bar k}}
    \end{pmatrix} 
\end{align*}
with $\hat k< k$ and each cell $\bar \rho_{\mu+\ell} = \{\rho_{\mu+\ell_1}, \dots, \rho_{\mu+\ell_i} \}$ with $\bar c_\ell = |\bar \rho_{\mu+\ell} |$ being the union of $\ell_i$ cells of $\pi^\star_\mathrm{un}$ in \eqref{eq:AEP_star}. If that is true, one gets that for $q \in \{\mathrm{f}, \mathrm{j}\}$
\begin{align*}
    M_q p(\bar \rho_{\mu+\ell}) = \sum_{s = 1}^{\bar k} \bar \varepsilon_{q,s} p(\bar \rho_{\mu+s}) = \begin{pmatrix}
    \bar \varepsilon_{q,1} \ones_{\bar c_1}\\
    \vdots \\
    \bar \varepsilon_{q,\bar k} \ones_{\bar c_{\bar k}}
    \end{pmatrix}.
\end{align*}
Accordingly, because $M_\mathrm{un} + M_\text{int} = M_\mathrm{f} + M_\mathrm{j}$ one gets that 
\begin{align*}
    (M_\mathrm{un} + M_\text{int})p(\bar \rho_{\mu+\ell}) =  \begin{pmatrix}
    (\bar \varepsilon_{\mathrm{f},1} + \bar \varepsilon_{\mathrm{j},1}) \ones_{\bar c_1}\\
    \vdots \\
    (\bar \varepsilon_{\mathrm{f},\bar k} + \bar \varepsilon_{\mathrm{j},1\bar k}) \ones_{\bar c_{\bar k}}
    \end{pmatrix}
\end{align*}
which implies that, by definition of intersection graph $\mathcal{G}_\text{int}$ and because $\mathcal{E}_\text{int} \subseteq \mathcal{E}_\mathrm{un}$, the subspace $\bar{\mathcal{P}}_c$ is $M_\text{int}$-invariant. However, because $\mathcal{P}_c$ is invariant as well and associated with the coarsest non-trivial AEP of $\mathcal{G}_\text{int}$ restricted to $\mathcal{C}_\mathrm{un}$ either $\bar k = k$ or $\bar k = 1$ thus falling into a contradiction and, thus, $\pi^\star_\mathcal{H} = \pi^\star_\mathrm{un}$. $\hfill$
\end{proof}

\begin{remark}
For the hybrid network, multi-consensus is fixed by both the union and intersection graphs $\mathcal{G}_\mathrm{un}$ and $\mathcal{G}_\text{int}$. As a matter of fact, cells in $\pi^\star_\mathcal{H}$ are given by the exclusive reaches of $\mathcal{G}_\mathrm{un}$ plus a further partition of nodes in $\mathcal{C}_\mathrm{un}$ suitably partitioned so to guarantee $M_\text{int}$-invariance of the corresponding characteristic matrix. In other words, once the first $\mu$ cells of $\pi^\star_\mathcal{H}$ are fixed by the reaches of the union graph, the remaining ones are defined by looking at the subgraph of $\mathcal{G}_\mathrm{int}$ arising from $M_{\mathrm{int}}$.
\end{remark}

\begin{remark}\label{rmk:4}
Whenever $M_\text{int} = 0$, multi-consensus is completely fixed by the union graph; namely, cells  $\rho_{\mu+\ell}$ are such that they provide the coarsest AEP associated with $\mathcal{G}_\mathrm{un}$.
\end{remark}

\begin{remark}\label{rmk:5}
Whenever the union graph possesses no common part (i.e., $\mathcal{C}_\mathrm{un} = \emptyset$), the hybrid network exhibits exactly $\mu$ consensuses behaviors. As a consequence, if $\mathcal{G}_\mathrm{un}$ possesses one consensus only, then the hybrid network converges to a single consensus independently of the clusters of $\mathcal{G}_\mathrm{f}$ and $\mathcal{G}_\mathrm{j}$.
\end{remark}

\subsection{Convergence analysis and hybrid multi-consensuses}
\bigskip
Let us now address the problem of convergence of solutions of \eqref{eq:hybrid_sys}
to the multi-consensus subspace. For the sake of simplicity, the proof of convergence is given in the periodic case only. The guidelines for the extension to the non periodic case are then briefly discussed.
\smallskip\\
First, the following technical result is proved.
\begin{prop}\label{prop:discs}
{\it Let $\tau>0$ be fixed and consider the hybrid multi-agent system (\ref{eq:hybrid_sys}) with a $\tau$-periodic hybrid time domain. Let us define the reverse monodromy matrix
\begin{equation}\label{eq:reverse_monodromy}
H=e^{-L_\mathrm{f}\tau}(I-\alpha L_\mathrm{j})
\end{equation}
where $\alpha>0$ is any positive parameter satisfying the conditions:
\begin{align}
&\alpha<\min_{\lambda\in\sigma(L_{\mathrm{j}})\setminus\{0\}}\frac{2\mathrm{Re}(\lambda)}{|\lambda|^2}\label{eq:alpha1}\\
&\alpha<\min_{\lambda\in\sigma(L_{\mathrm{j}})\setminus\{0\}}\frac{1}{\mathrm{Re}(\lambda)}\label{eq:alpha2}
\end{align}    }

Then the spectrum of $H$ verifies
$$
\sigma(H)\subset\mathbb{C}_{\leq1}:=\{z\in\mathbb{C}:|z|< 1\}\cup\{1\}.
$$
\end{prop}
\begin{proof}
By direct inspection, as long as $\alpha$ is chosen within the prescribed range, it can be seen that the entries of the matrix $H$ satisfy
\begin{itemize}
 \item[(i)] $H_{ij}\geq0$ for any $i,j=1,...,N$
\item[(ii)] the sum of each row is equal to 1, i.e.,
$$
H_{i1}+H_{i2}+\cdots+H_{iN}=1
$$
    \item[(iii)] $H_{ii}\in(0,1]$ for any $i=1,...,N$
\end{itemize}
The first two items rely on the structure of the Laplacians $L_\mathrm{f},L_\mathrm{j}$, the positivity of $\tau$ and the choice of $\alpha$. Condition \eqref{eq:alpha1} guarantees that $(I-\alpha L_\mathrm{j})$ is Schur (see Proposition \ref{prop:spectrum}), whereas the additional condition \eqref{eq:alpha2} is useful to further restrict the spectrum within the half-plane $\mathbb{C}^+$. In particular, for any $\tau>0$ and thanks to \eqref{eq:alpha1}-\eqref{eq:alpha2}, the matrices $e^{-L_\mathrm{f}\tau}$ and $(I-\alpha L_\mathrm{j})$ are stochastic and are guaranteed to have positive entries on diagonal terms and non-negative entries elsewhere. Such properties are naturally retained in the product, showing that $H$ has non-negative entries (item~(i)) and is still a stochastic matrix (item~(ii)). Finally, combining these two conditions entails item~(iii). The conclusion then follows from the classical Geršgorin theorem \citep{golub2013matrix}. In fact, using the identity $\sum_{j\neq i}|H_{ij}|=1-H_{ii}$ which is guaranteed by the previous conditions, one has that the eigenvalues belong to the union of discs $\mathbb{D}_i$ with
$$
\mathbb{D}_i:=\{z\in\mathbb{C}: |z-H_{ii}|\leq(1-H_{ii})\}
$$
and clearly the desired inclusion $\bigcup_{i=1}^N\mathbb{D}_i\subset\mathbb{C}_{\leq 1}$ holds.
\end{proof}
\smallskip
Let us denote with $\mathcal{K} = \mathcal{P}^\star_\mathcal{H}\subset\mathbb{R}^N$ the multi-consensus subspace, this being in particular a center invariant subspace \citep{isidori2017lectures} for the monodromy matrix, associated with the eigenvalue~1. In view of Proposition~\ref{prop:discs}, a linear transformation $x=T\tilde{x}$ exists such that 
\begin{equation}\label{eq:transf}
\tilde{x}=[\tilde{x}_1^\top\quad\tilde{x}_2^\top]^\top,\quad
T^{-1}HT=\tilde{H}=\begin{bmatrix}\tilde{H}_{11}&0\\
0&\tilde{H}_{22}
\end{bmatrix}
\end{equation}
with $\sigma(\tilde{H}_{11})=\{1\}$ and  $\sigma(\tilde{H}_{22})\subset\mathbb{C}_{<1}$. In particular, in the new coordinates, the multi-consensus space is represented by the set $$\tilde{\mathcal{K}}=T^{-1}\mathcal{K}=\{\tilde{x}\in\mathbb{R}^N: \tilde{x}_2=0\}.$$ By the  Lyapunov stability theorem for discrete-time linear systems \citep[Theorem 5.D5]{chen1999linear}, a scalar $\kappa>0$ and a positive definite matrix $$\tilde{P}_{22}=\tilde{P}_{22}^\top\succ0,\quad  \tilde{P}_{22}\in\mathbb{R}^{N_2\times N_2}$$ can be found such that
$$
\tilde{x}_2^\top\left(\tilde{H}^\top_{22}\tilde{P}_{22}\tilde{H}_{22}-\tilde{P}_{22}\right)\tilde{x}_2\leq-\kappa  \|\tilde{x}_2\|^2
\quad \forall \tilde{x}_2\in\mathbb{R}^{N_2}$$  with $N_2=N-\dim(\mathcal{K})$. Furthermore, setting \begin{equation}\label{eq:P1}
\tilde{P}=\begin{bmatrix}
0&0\\
0&\tilde{P}_{22}
\end{bmatrix}\succeq0,
\end{equation}
the inequality can be extended on the whole state space $\mathbb{R}^N$ as
\begin{equation}\label{eq:P2}
\begin{array}{rl}
\tilde{x}^\top(\tilde{H}^\top\tilde{P}\tilde{H}-\tilde{P})\tilde{x}&\leq-\kappa\tilde{x}^\top\begin{bmatrix}
0&0\\
0&I_{N_2\times N_2}\end{bmatrix}\tilde{x}\smallskip\\
&=-\kappa\,(\mathrm{dist}(\tilde{x},\tilde{\mathcal{K}}))^2
\end{array}\quad\forall \tilde{x}\in\mathbb{R}^N
\end{equation}
Now, by defining the augmented state $\zeta=[\tilde{x}^\top \ s]^\top \in\mathbb{R}^{N+1}$ and considering the flow and jump sets 
\begin{align}\label{eq:domains}
\begin{array}{c}
\mathcal{C}=\{\zeta=(\tilde{x},s): \tilde{x}\in\mathbb{R}^N, s\in[0,\tau]\}\smallskip\\
\mathcal{D}=\{\zeta=(\tilde{x},s): \tilde{x}\in\mathbb{R}^N, s=\tau\},
\end{array}
\end{align}
the hybrid-multi agent system (\ref{eq:hybrid_sys}) can be rewritten as
\begin{equation}\label{eq:hybrid_equiv}
\begin{array}{ll}
 \dot\zeta=F(\zeta)    &  \zeta\in\mathcal{C}\smallskip\\
\zeta^+=G(\zeta)     & \zeta\in\mathcal{D}
\end{array}    
\end{equation}
where the flow and jump maps are defined, respectively, by
$$
F(\zeta)=\begin{bmatrix}
-\tilde{L}_\mathrm{f}\tilde{x}\\
1
\end{bmatrix}\quad G(\zeta)=\begin{bmatrix}
(I-\alpha \tilde{L}_\mathrm{j})\tilde{x}\\
0
\end{bmatrix}
$$
with $\tilde{L}_\mathrm{f}=T^{-1}L_\mathrm{f}T$ and $\tilde{L}_{\mathrm{j}}=T^{-1}L_{\mathrm{j}}T$.
Based on the latter arguments, the following stability result can be established.
\begin{theorem}\label{th:conv}
{\it Let $\tau>0$ be fixed and consider the hybrid multi-agent system (\ref{eq:hybrid_sys}) with a $\tau$-periodic hybrid time domain and {$\displaystyle\alpha>0$ satisfying \eqref{eq:alpha1}-\eqref{eq:alpha2}.} The multi-consensus subspace $\mathcal{K}$ is globally asymptotically stable in the hybrid sense \citep[~Definition 3.6]{goebel2012hybrid}.}
\end{theorem}

\smallskip
\begin{proof}
We first observe that, thanks to the linearity of the transformation, proving the asymptotic stability of the set $\mathcal{K}$ in the original coordinates $x$ is equivalent to establishing asymptotic stability of $\tilde{\mathcal{K}}=T^{-1}\mathcal{K}$ in the new coordinates $\tilde{x}$.
Now, following \citep{goebel2012hybrid} and referring to the equivalent formulation (\ref{eq:hybrid_equiv}), we  aim at showing that the closed set $\mathcal{A}=\{\zeta: \tilde{x}=(\tilde{x}_1,0),\ s\in[0,\tau]\}=\tilde{\mathcal{K}}\times[0,\tau]$ is globally asymptotically stable in the hybrid sense \citep[~Definition 3.6]{goebel2012hybrid}. To this end, we consider the Lyapunov function candidate $$
V(\zeta)=e^{-\beta s}W(e^{-\tilde{L}_\mathrm{f}(\tau-s)}\tilde{x})
$$
with $\beta>0$ and $W(\tilde{x})=\tilde{x}^\top \tilde{P}\tilde{x}$, where $\tilde{P}=\tilde{P}^\top \succeq0$ is defined in (\ref{eq:P1})-(\ref{eq:P2}).
Observing that by Theorem \ref{th:AEP_c} the inclusion $\tilde{L}_\mathrm{f}\tilde{\mathcal{K}}\subseteq\tilde{\mathcal{K}}$ holds, we can infer that the exponential matrix $e^{-\tilde{L}_\mathrm{f}(\tau-s)}$ must have a upper block-triangular form, with
$$
e^{-\tilde{L}_\mathrm{f}(\tau-s)}=\begin{bmatrix}
\mathcal{E}_{11}(\tau-s)&\mathcal{E}_{12}(\tau-s)\\
0&\mathcal{E}_{22}(\tau-s)
\end{bmatrix}
$$
for some suitable matrix-valued continuous functions $\mathcal{E}_{11}(s),\mathcal{E}_{12}(s)$ and $\mathcal{E}_{22}(s)$ where $\det(\mathcal{E}_{22}(s))\neq0\ \forall s\in\mathbb{R}$. Based on this fact and on the diagonal structure of $\tilde{P}$, the following identity holds:
$$e^{-\tilde{L}^\top_\mathrm{f}(\tau-s)}\tilde{P}e^{-\tilde{L}_\mathrm{f}(\tau-s)}=\begin{bmatrix}
0&0\\
0&\mathcal{E}_{22}^\top(\tau-s)\tilde{P}_{22}\mathcal{E}_{22}(\tau-s)\end{bmatrix}.$$
Now, setting
$$
\|\zeta\|_{\mathcal{A}}=\mathrm{dist}(\zeta,\mathcal{A})=\inf_{\upsilon \in\mathcal{A}}\|\zeta-\upsilon\|,
$$ 
it can be easily verified that the bounds $$
c_1 \|\zeta\|_{\mathcal{A}}^2\leq V(\zeta)\leq c_2 \|\zeta\|_{\mathcal{A}}^2
$$
are fulfilled for any $\zeta\in\mathcal{C}\cup\mathcal{D}$ with positive constants $c_1,c_2>0$ given by
$$
\begin{array}{c}
c_1=e^{-\beta\tau}\lambda_{\min}(\tilde{P}_{22})\min_{s\in[0,\tau]}\lambda_{\min}(\mathcal{E}_{22}^\top(\tau-s)\mathcal{E}_{22}(\tau-s)),\quad \smallskip\\
c_2=\lambda_{\max}(\tilde{P}_{22})\max_{s\in[0,\tau]}\lambda_{\max}(\mathcal{E}_{22}^\top(\tau-s)\mathcal{E}_{22}(\tau-s)).
\end{array}
$$
Computing the derivative along the solution of (\ref{eq:hybrid_equiv}) during any flow interval yields
$$
\begin{array}{ll}\langle\nabla_{\zeta}V(\zeta),F(\zeta)\rangle&= -2e^{-\beta s}\tilde{x}^\top e^{-\tilde{L}_\mathrm{f}^\top (\tau-s)}\tilde{P}e^{-\tilde{L}_\mathrm{f}(\tau-s)}\tilde{L}_\mathrm{f}\tilde{x}\\
&+\,2e^{-\beta s}\tilde{x}^\top e^{-\tilde{L}_\mathrm{f}^\top (\tau-s)}\tilde{P}e^{-\tilde{L}_\mathrm{f}(\tau-s)}\tilde{L}_\mathrm{f}\tilde{x}\\
&-\beta e^{-\beta s}W(e^{-\tilde{L}_\mathrm{f}(\tau-s)}\tilde{x})\\
&\leq-\beta c_1 \|\zeta\|_{\mathcal{A}}^2
\end{array}
$$
for any $\zeta\in\mathcal{C}$. On the other hand, for $\zeta\in\mathcal{D}$, one has
$$
\begin{array}{ll}
V(G(\zeta))-V(\zeta)\!\!\!\!\!&=\tilde{x}^\top (\tilde{H}^\top \tilde{P}\tilde{H}-e^{-\beta \tau}\tilde{P})\tilde{x}\\
&\leq(-\kappa+(1-e^{-\beta \tau})\lambda_{\max}(\tilde{P}_{22}))\|\zeta\|_{\mathcal{A}}^2
\end{array}
$$
with $\tilde{H}$ defined in (\ref{eq:reverse_monodromy})-(\ref{eq:transf}) and where (\ref{eq:P2}) has been used. Selecting $\beta>0$ sufficiently small, i.e.\footnote{Note that by construction $0<\kappa<\lambda_{\max}(\tilde{P}_{22})$.},
$$
\beta<\frac{1}{\tau}\log\left(\frac{\lambda_{\max}(\tilde{P}_{22})}{\lambda_{\max}(\tilde{P}_{22})-\kappa}\right),
$$
the right-hand side of $V(G(\zeta))-V(\zeta)$ is negative definite relative to $\mathcal{A}$ and this is enough to guarantee the hybrid asymptotic stability of the set $\mathcal{A}$ by invoking \citep[Theorem 3.18]{goebel2012hybrid}. In conclusion, the convergence of $x$ towards the multi-consensus space $\mathcal{K}$ has been proved.
\end{proof}

\smallskip
\begin{remark}
The extension to the non-periodic case can be done by considering an augmented state $\zeta = [\tilde{x}^\top  \ j \ s ]^\top $ and replacing the flow and jump sets in (\ref{eq:domains}) by
\begin{align*}
\begin{array}{c}
\mathcal{C}=\{\zeta=(\tilde{x},j, s): x\in\mathbb{R}^N, j\in \mathbb{N}, s\in[0, t_{j+1}-t_j]\}\smallskip\\
\mathcal{D}=\{\zeta=(\tilde{x},j, s): x\in\mathbb{R}^N, j \in \mathbb{N}, s=t_{j+1}-t_j \}
\end{array}
\end{align*}
where we implicitly assume $t_0 = 0$. Note that both $\mathcal{C}$ and $\mathcal{D}$ are closed sets thanks to the existence of lower and upper bounds for the dwell time (i.e., $\tau_\text{min}< t_{j+1}-t_j <\tau_\text{max}$). See \citep[~Example 3.22]{goebel2012hybrid} for an explicit definition of the Lyapunov function.
\end{remark}

From the results above, it is possible to characterize multi-consensus of the hybrid network (\ref{eq:hybrid_sys}). 
To this end, assuming the concerned Laplacians of the form (\ref{Lq}) and the definition of $\pi_{\mathcal{H}}^\star = \{\mathcal{H}_{\mathrm{un},1}, \dots, \mathcal{H}_{\mathrm{un},\mu}, \rho_{\mu+1}, \dots, \rho_{\mu+k}\}$ with $| \rho_{\mu+\ell}| = c_\ell$ in Theorem (\ref{th:AEP_c}), we denote for $i = 1, \dots, \mu$ and $\ell = 1, \dots, k$
\begin{align*}
    \mathbf{x}_i =& \begin{pmatrix}x_{h_1+\dots + h_{i-1}+1} & \dots & x_{h_1+\dots+h_{i-1}+h_i} \end{pmatrix}^\top\\
    \mathbf{x}_{\delta}^r =& \begin{pmatrix}x_{N-\mu+c_1 +\dots +  c_{i-1}+1} & \dots & x_{N-\mu+c_1 + \dots +c_i} \end{pmatrix}^\top 
\end{align*}

\begin{cor}\label{cor:1}
Consider the hybrid multi-agent system (\ref{eq:hybrid_sys}) with $\mathcal{G}_\mathrm{f}$ and $\mathcal{G}_\mathrm{j}$ the flow and jump communication graphs characterized by the Laplacians $L_\mathrm{f}$ and $L_\mathrm{j}$ respectively. Let $\mathcal{G}_\mathrm{un}$ be union graph with Laplacian of the form (\ref{Lu}) with reaches $\mathcal{H}_{\mathrm{un},i}$ of cardinality $h_i$ ($i =1, \dots, \mu$) and common component $\mathcal{C}_\mathrm{un}$ of cardinality $\delta$. Consider the weighted Laplacian
\begin{align*}
    L_\alpha = L_\mathrm{f} + \alpha L_\mathrm{j}
\end{align*} 
of the same form as (\ref{Lu}) with, setting $h_i = |\mathcal{H}_{\mathrm{un},i}| $ for $i = 1, \dots, \mu$, right  and left eigenvectors associated with $\lambda = 0$ given by 
\begin{align}\label{eig:L_a}
&z_{\alpha, 1} = \begin{pmatrix}
\ones_{h_1} \\ 
\vdots \\
0\\
\gamma_{\alpha}^1
\end{pmatrix} \quad \dots \quad z_{\alpha,\mu} = \begin{pmatrix}
0 \\ 
\vdots \\
\ones_{h_\mu}\\
\gamma_{\alpha}^\mu
\end{pmatrix}\\ 
\label{left_eig_a}
   & \tilde v_{\alpha,1}^\top = \begin{pmatrix} v_{\alpha,1}^\top & \dots & \mathbf{0}
    \end{pmatrix}\quad \dots \quad \tilde v_{\alpha,\mu}^\top = \begin{pmatrix} \mathbf{0} & \dots & v_{\alpha,\mu}^\top 
    \end{pmatrix}
\end{align}
with $v_{\alpha,i}^\top = \begin{pmatrix} v_{\alpha,i}^1 & \dots & v_{\alpha,i}^{h_i} \end{pmatrix} \in \mathbb{R}^{1\times h_i}$, $v_i^s>0$. 
Let $\pi^\star_\mathcal{H}$ be the coarsest AEP for both $\mathcal{G}_\mathrm{f}$ and $\mathcal{G}_\mathrm{j}$ as defined in Theorem \ref{th:AEP_c}.
Then, the following holds true:
\begin{description}
\item[(i)] for all nodes in the same exclusive reach $\mathcal{H}_{\mathrm{un},i}$ of $\mathcal{G}_\mathrm{un}$, $\mathbf{x}_i(t,j) \to x_i^{ss}\ones_{h_i}$ with $x_{i}^{ss} =  v_{\alpha,i}^\top \mathbf{x}_i(t_0,0)$ with $i = 1, \dots, \mu$; 
\item[(ii)] all nodes in the same cell $\rho_{\mu+r} \subset \mathcal{C}_\mathrm{un}$ with $r = 1, \dots, k$ converge to the hybrid consensus dynamics
\begin{align*}
    \dot{\mathbf{x}}_{\delta,r}^{ss} =& - M_{\mathrm{f},rr}{\mathbf{x}}_{\delta,r}^{ss} -  \sum_{i = 1}^{\mu}\varepsilon^i_{\mathrm{f},r} x_{i}^{ss} \ones_{h_i}  - \sum_{\ell = 1, \ell \neq r}^{k} M_{\mathrm{f},r\ell}{\mathbf{x}}_{\delta,\ell}^{ss}\\
     {\mathbf{x}_{\delta,r}^{ss}}^+ =& (I_{c_r}-\alpha M_{\mathrm{j},rr}){\mathbf{x}}_\delta^r  -  \sum_{i = 1}^{\mu}\varepsilon^i_{\mathrm{j},r} x_{i}^{ss} \ones_{h_i}  - \sum_{\ell = 1, \ell \neq r}^{k} M_{\mathrm{j},r\ell}{\mathbf{x}}_{\delta,\ell}^{ss}.   
\end{align*}
\item[(iii)] if $\Gamma^\mu_\alpha = \text{span}\{\gamma_{\alpha}^1, \dots, \gamma_{\alpha}^\mu\}$ is both $M_\mathrm{f}$ and $M_\mathrm{j}$ invariant, then nodes belonging to the same cell $\rho_{\mu+r}$ with $r = 1, \dots, k$ converge to $k>0$ constant multi-consensuses; namely, for all $r = 1, \dots, k$ splitting $\gamma^i = \begin{pmatrix}
\gamma^i_{\alpha,1} \ones_{c_1} & \dots & \gamma^i_{\alpha,k} \ones_{c_k}
\end{pmatrix}^\top
$ one gets
\begin{align*}
    \mathbf{x}_{\delta,r}(t,k) \to \sum \gamma^i_{\alpha,r}  x_i^{ss} \ones_{c_i}.
\end{align*}
\end{description}
\end{cor}

\smallskip
\begin{proof}
The proof follows from the triangular form of the involved Laplacians (\ref{Lq}) and Theorems \ref{th:AEP_c} and \ref{th:conv}.  
\end{proof}

As stated in the result above, consensus in hybrid networks (even when composed of scalar integrators) is a hybrid trajectory: nodes in the reaches tend to a constant consensus value given by the weighted mean of the corresponding agent's initial states; nodes in the common tend to a hybrid consensus arc which is however bounded. 

\begin{remark}
The multi-consensuses of the hybrid network are all parameterized by the coupling strength $\alpha>0$ in (\ref{eq:hybrid_sys}) and are independent, as expected, of the jumping times and flow periods. 
\end{remark}

\section{Examples}\label{sec:example}
In this section, different networks composed of $N = 7$ hybrid agents of the form (\ref{eq:hybrid_agent}) are considered. In all cases, the initial conditions are set as $x_i(0, 0) = i$ for $i = 1, \dots, 7$ with aperiodic jumps with $\tau_\text{min} = 0.1$ and $\tau_\text{max} = 1$ seconds. For completeness, the behaviors of the purely continuous-time networks evolving under either $\mathcal{G}_\mathrm{f}$ and $\mathcal{G}_\mathrm{j}$ are reported to highlight the influence of the hybrid connection compared to the standalone ones. In all cases, the coupling strength for the jump component of the dynamics in \ref{eq:inputs} is set to $\alpha = \frac{1}{5}.$
\begin{example}
Let us consider graphs as in Fig.~\ref{fig:ex1}. 
In this case, applying the result in Theorem \ref{th:AEP_c} and along the lines of Remark \ref{rmk:5}, the union graph the AEP underlying consensus is given by the trivial one composed of one cell with all nodes included; namely, one obtains
\begin{align*}
    \pi^\star_\mathcal{H} =& \{ \{\nu_1,\nu_2,\nu_3,\nu_5, \nu_5, \nu_6, \nu_7\}\}.
\end{align*}
As a consequence, although the standalone flow and jump graphs induce three and two consensuses respectively with
\begin{align*}
    \pi^\star_\mathrm{f} =& \{ \{\nu_1,\nu_2,\nu_3 \}, \{\nu_4,\nu_5 \}, \{\nu_6,\nu_7\} \} \\
    \pi^\star_\mathrm{j} =& \{ \{\nu_4,\nu_7 \}, \{\nu_1, \nu_2, \nu_3, \nu_5, \nu_6\} \},
\end{align*} nodes of the hybrid network converge to a unique asymptotic agreement which is given by
\begin{align*}
    x^{ss} =& v^\top_\alpha x_0\\ v^\top_\alpha =&\frac{1}{16} \begin{pmatrix}
   0.1 & 0.12& 0.66 & 0.051& 0.048& 0.012& 6.8e^{-3}
    \end{pmatrix}
\end{align*}
with $v^\top_\alpha$ being the left normalized eigenvector associated with the zero eigenvalue of the Laplacian of the union graph $\mathcal{G}_\mathrm{un}$.
More in detail, as underlined by Figure \ref{fig:sim_ex1}, in this case one gets the constant consensus value $x^{ss} = 2.87$.
\begin{figure}[h!]
\centering
\begin{subfigure}[b]{0.49\columnwidth}
\centering
\includegraphics[width=1\textwidth]{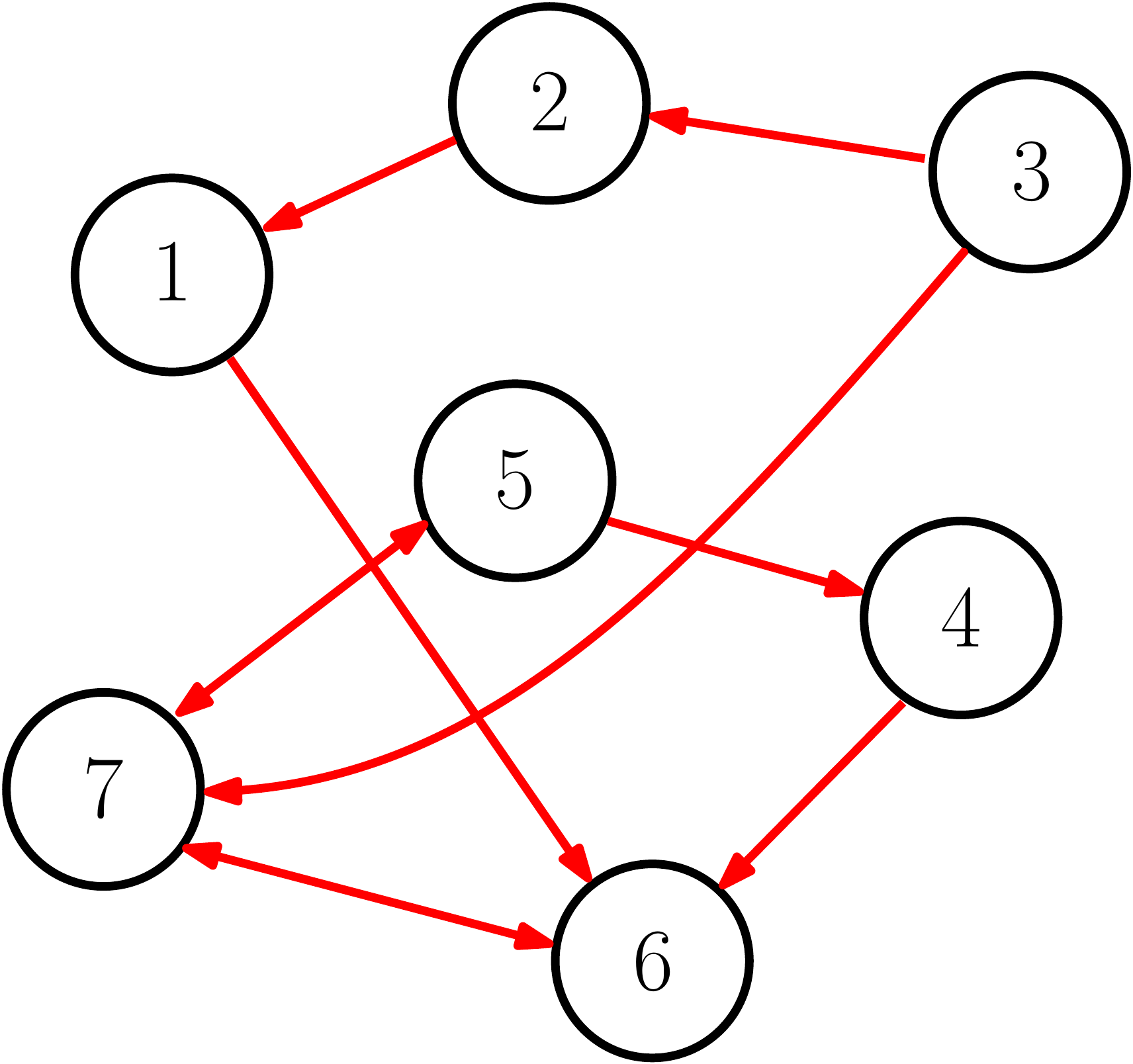}
\caption{Flow graph}
\label{fig:ex1_flow}
\end{subfigure}
\hfill
\begin{subfigure}[b]{0.49\columnwidth}
\centering
\includegraphics[width=1\textwidth]{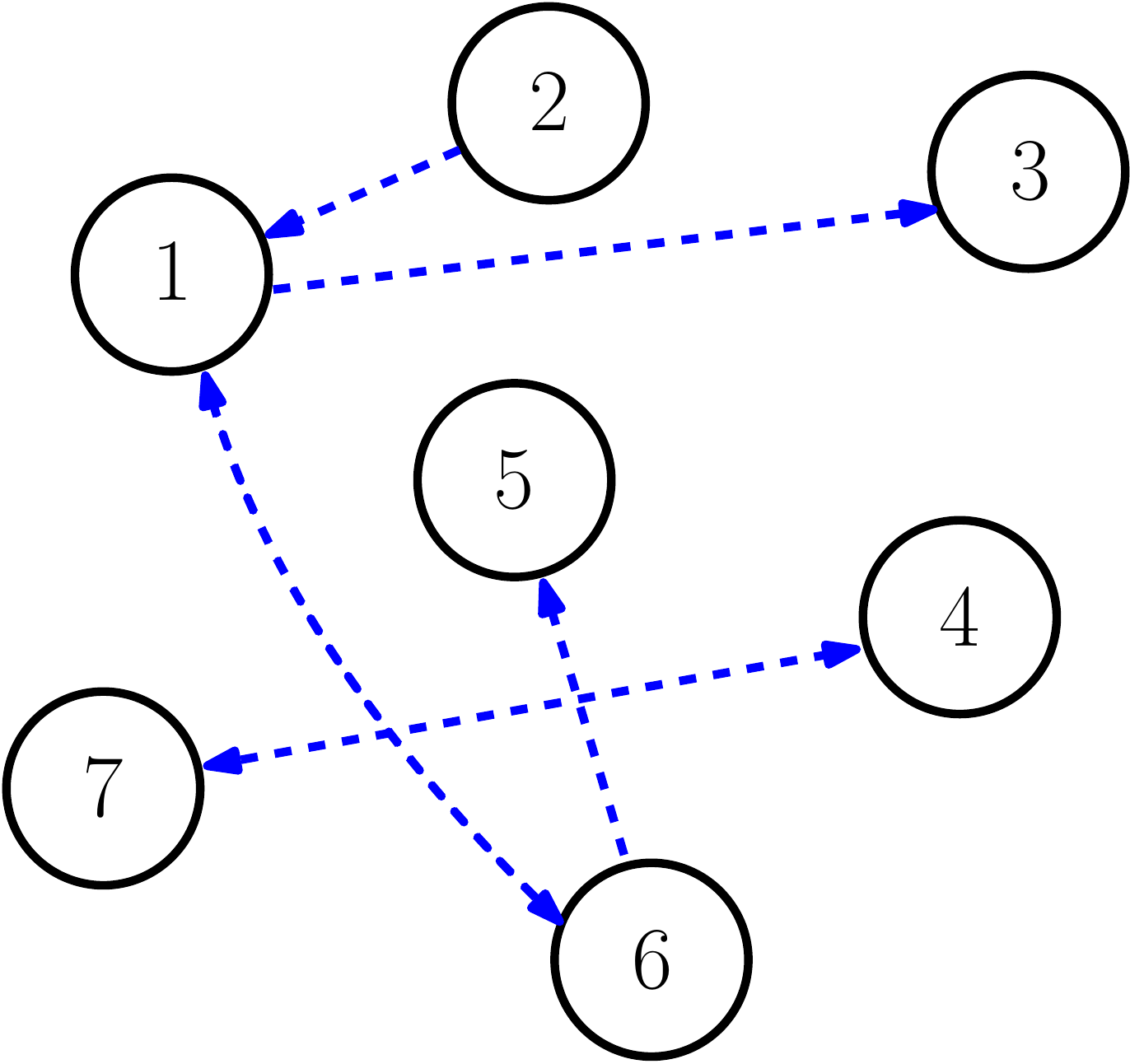}
\caption{Jump graph}
\label{fig:ex1_jump}
\end{subfigure}
\caption{Graphs for Example~1}\label{fig:ex1}
\end{figure}
\begin{figure}[h!]
\centering
\includegraphics[width=\textwidth]{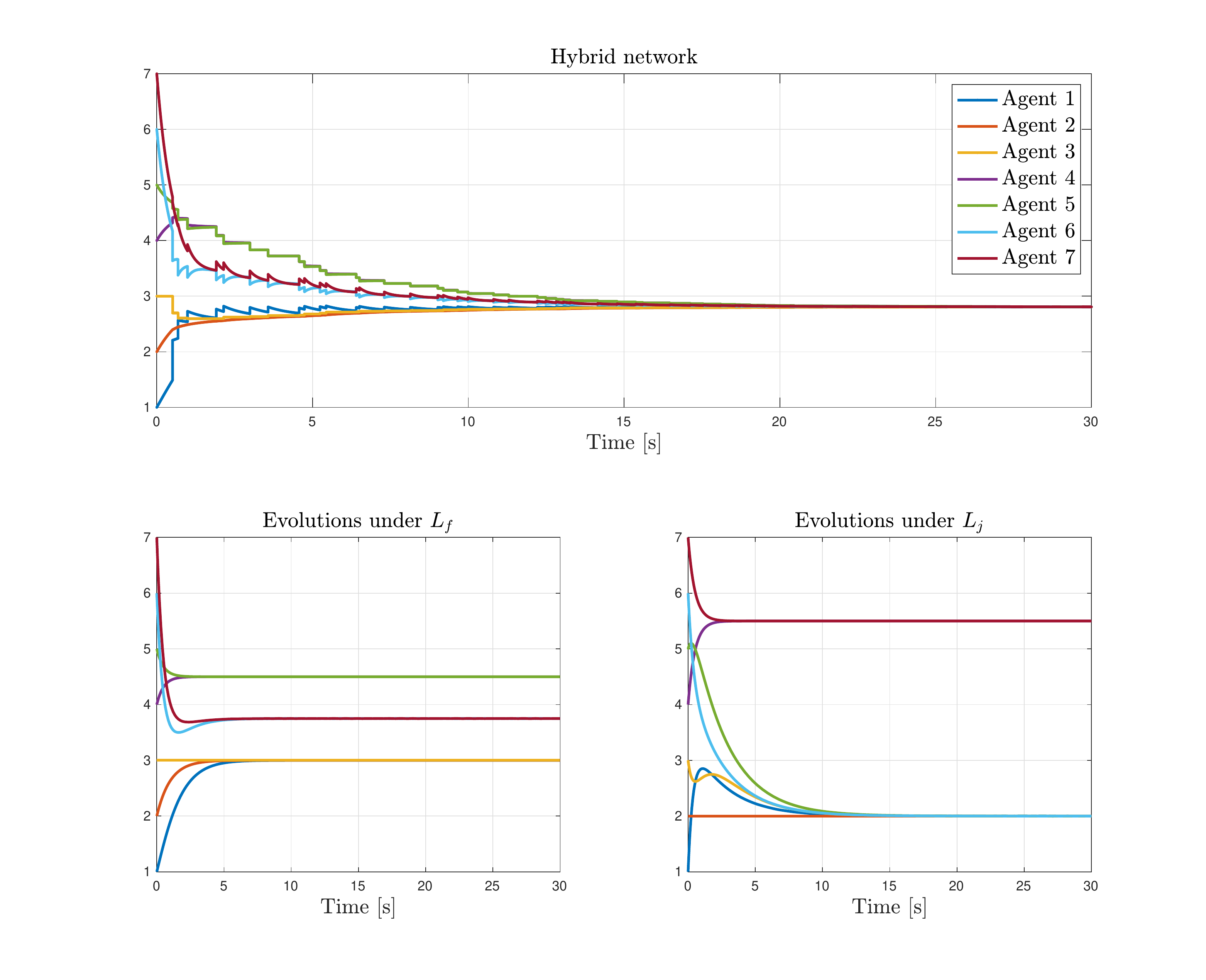}
\caption{Simulations for Example 1}
\hfill
\label{fig:sim_ex1}
\end{figure}
\end{example}

\begin{example}
\begin{figure}[h!]
\centering
\begin{subfigure}[b]{0.49\columnwidth}
\centering
\includegraphics[width=1\textwidth]{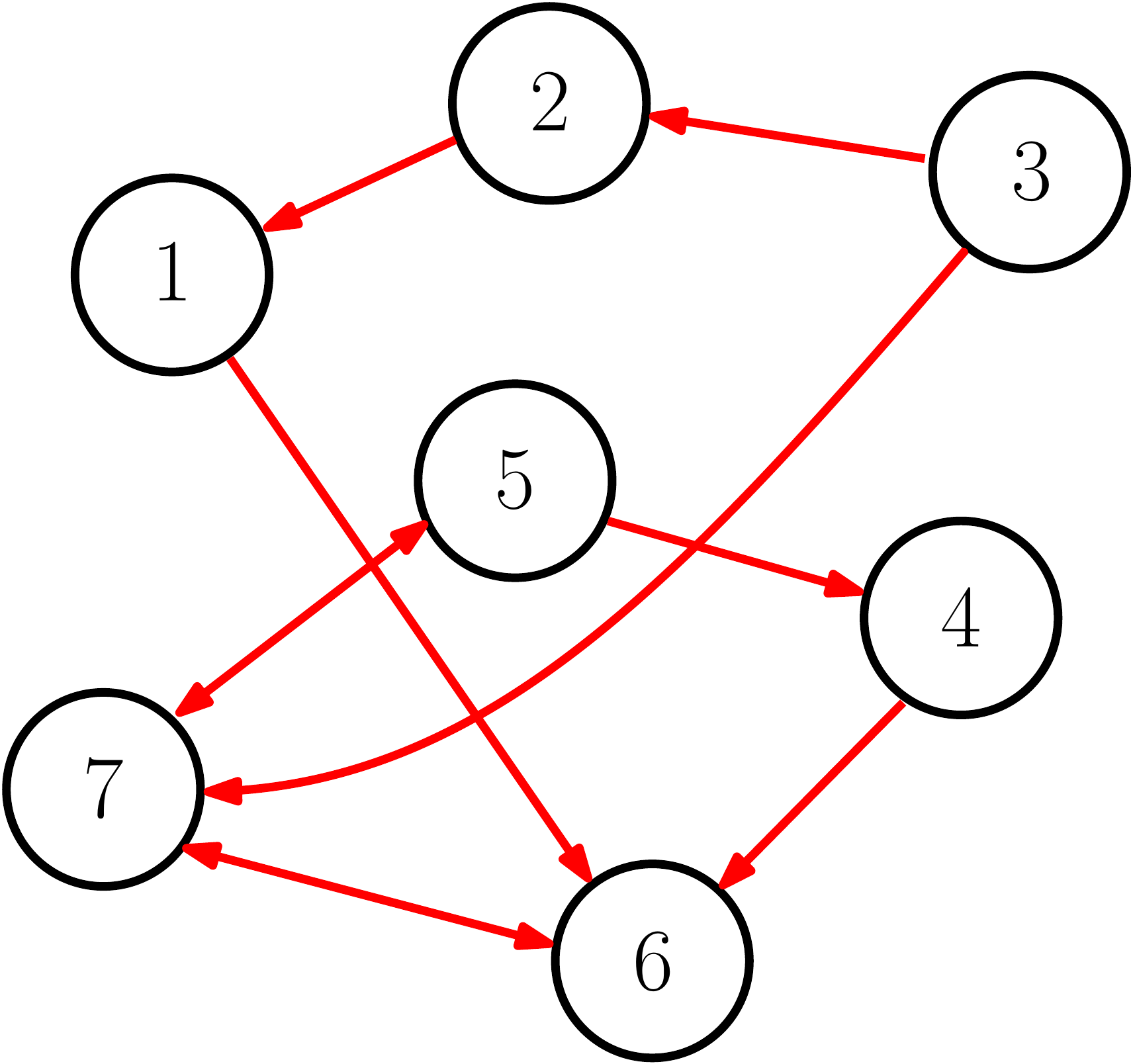}
\caption{Flow graph}
\label{fig:ex2_flow}
\end{subfigure}
\hfill
\begin{subfigure}[b]{0.49\columnwidth}
\centering
\includegraphics[width=1\textwidth]{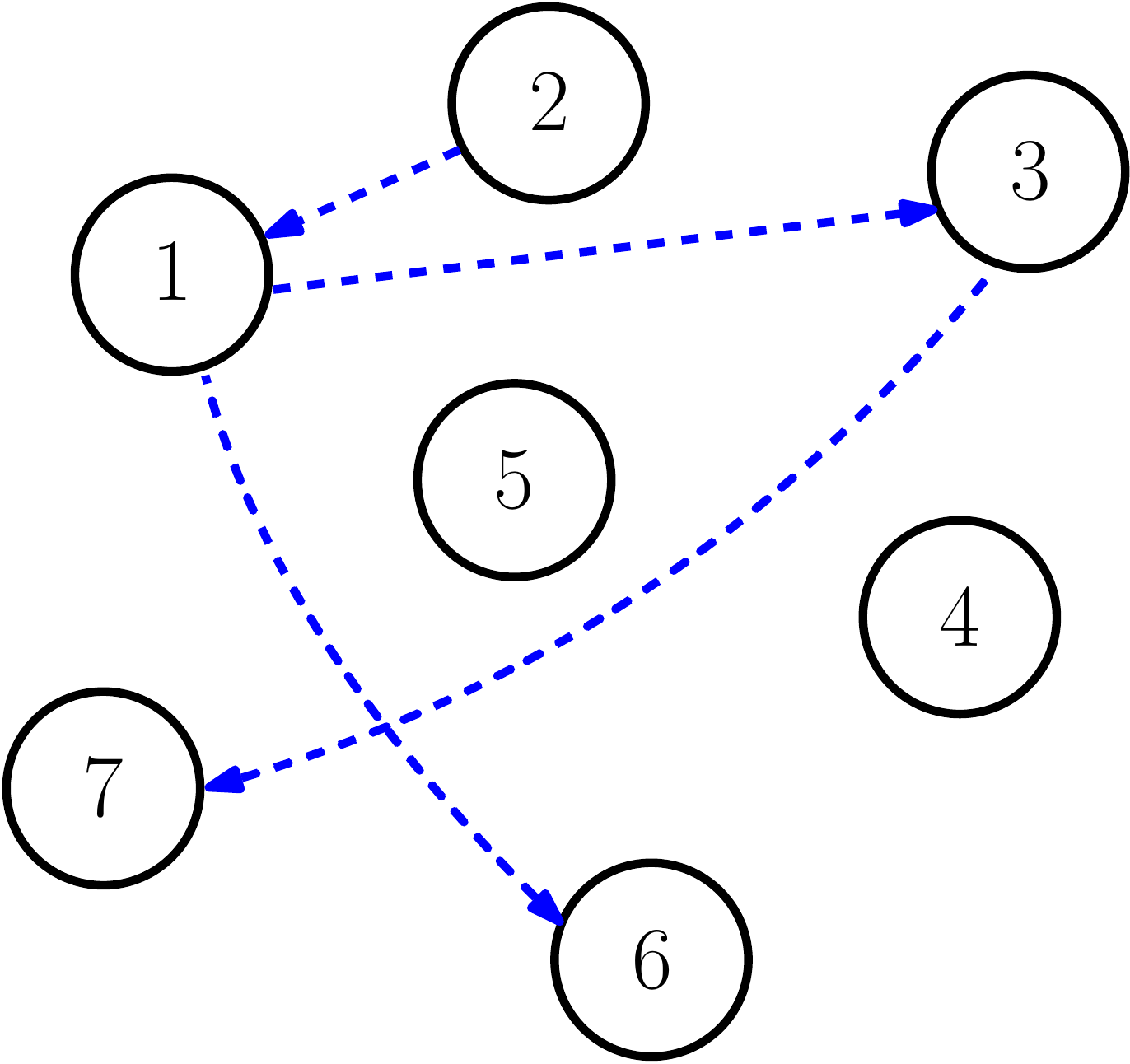}
\caption{Jump graph}
\label{fig:ex2_jump}
\end{subfigure}
\caption{Graphs for Example~2}\label{fig:ex2}
\end{figure}
Let us consider graphs as in Fig.~\ref{fig:ex2}. In this case, invoking Theorem \ref{th:AEP_c}, nodes of the hybrid network cluster according to the AEP
\begin{align*}
    \pi^\star_\mathcal{H} =\{
    \mathcal{H}_{\mathrm{un},1},
    \mathcal{H}_{\mathrm{un},2}, \rho_{3}\}
\end{align*}
with the exclusive reaches of the union graph $\mathcal{H}_{\mathrm{un},1} = \{\nu_1, \nu_2, \nu_3 \}$, $\mathcal{H}_{\mathrm{un},2} = \{\nu_4, \nu_5\}$ and $\rho_3 = \mathcal{C}_\mathrm{un} = \{\nu_6,\nu_7 \}$ associated with $
M_\text{int} = \mathbf{0}$. In addition, according to Corollary \ref{cor:1}, nodes in the exclusive reaches converge to the consensuses 
\begin{align*}
    x_{1}^{ss} =& v_{\alpha,1}^\top \mathbf{x}_1(0, 0), \quad v_{\alpha,1}^\top = \frac{30}{41}\begin{pmatrix} \frac{1}{6} & \frac{1}{5} & 1 \end{pmatrix}\\
    x_{2}^{ss} =& v_{\alpha,2}^\top \mathbf{x}_2(0, 0), \quad v_{\alpha,2}^\top = \frac{1}{2}\begin{pmatrix} 1 & 1 \end{pmatrix}
\end{align*}
so getting $x_1^{ss} = 2.6098$ and $x_2^{ss} = 4.5$. On the other side, as shown in Figure \ref{fig:sim_ex2}, nodes in the cell associated with the common component converge to the same hybrid consensus arc described by
\begin{align*}
    \dot{\mathbf{x}}_\delta =& -\begin{pmatrix}
    3 & -1 \\ -1 & 3
    \end{pmatrix}{\mathbf{x}}_\delta-\begin{pmatrix}
    1 \\ 1
    \end{pmatrix} x_{1}^{ss} - \begin{pmatrix}
    1 \\ 1
    \end{pmatrix} x_{2}^{ss} \\
    {\mathbf{x}}_\delta^+ =& -\begin{pmatrix}
    1-\alpha & 0 \\ 0 & 1-\alpha
    \end{pmatrix}{\mathbf{x}}_\delta-\begin{pmatrix}
    1 \\ 1
    \end{pmatrix} x_{1}^{ss}.
\end{align*}
\begin{figure}[h!]
\centering
\includegraphics[width=\textwidth]{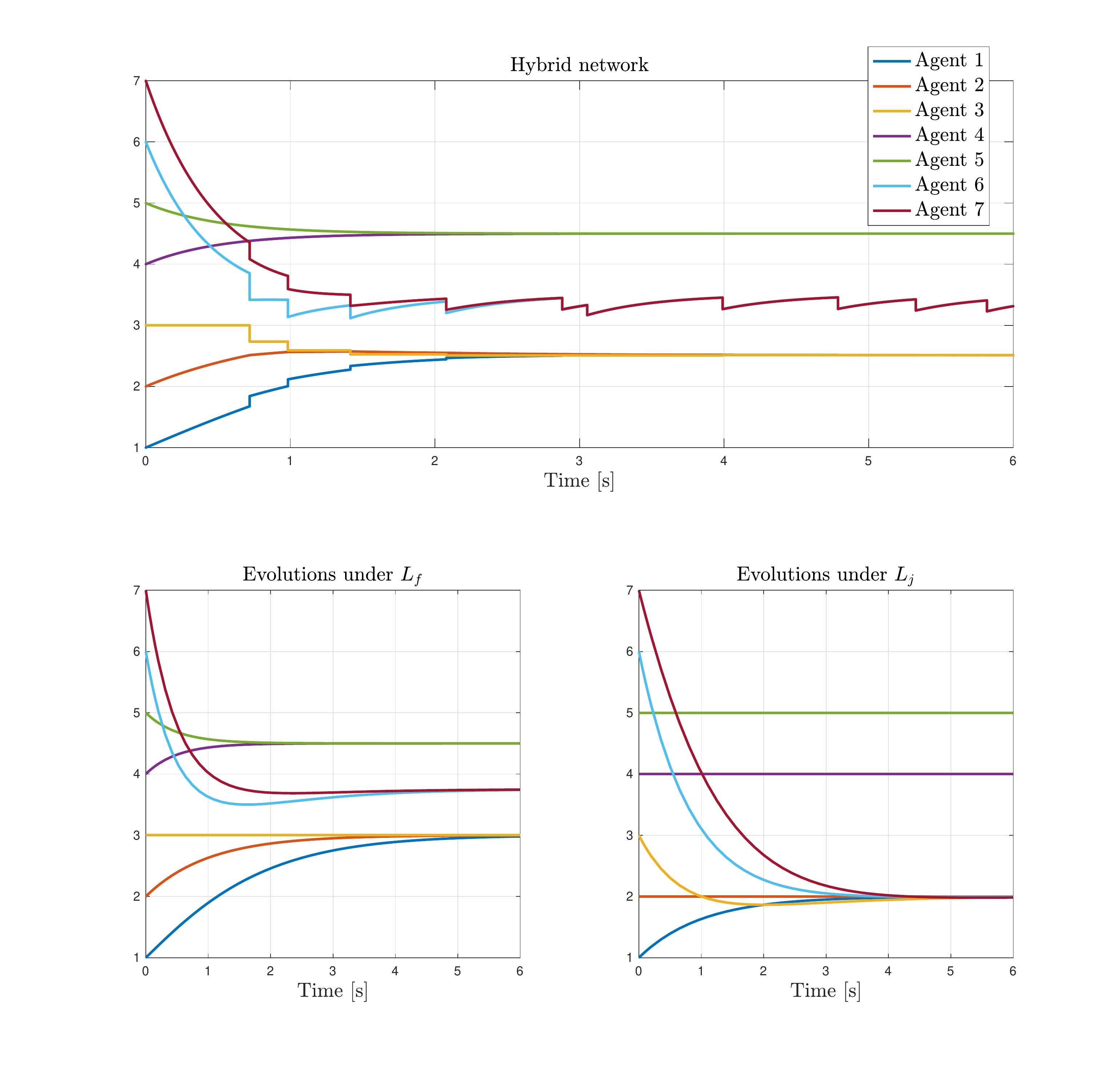}
\caption{Simulations for Example 2}
\hfill
\label{fig:sim_ex2}
\end{figure}
\end{example}

\begin{example}

\begin{figure}[h!]
\centering
\begin{subfigure}[b]{0.49\columnwidth}
\centering
\includegraphics[width=1\textwidth]{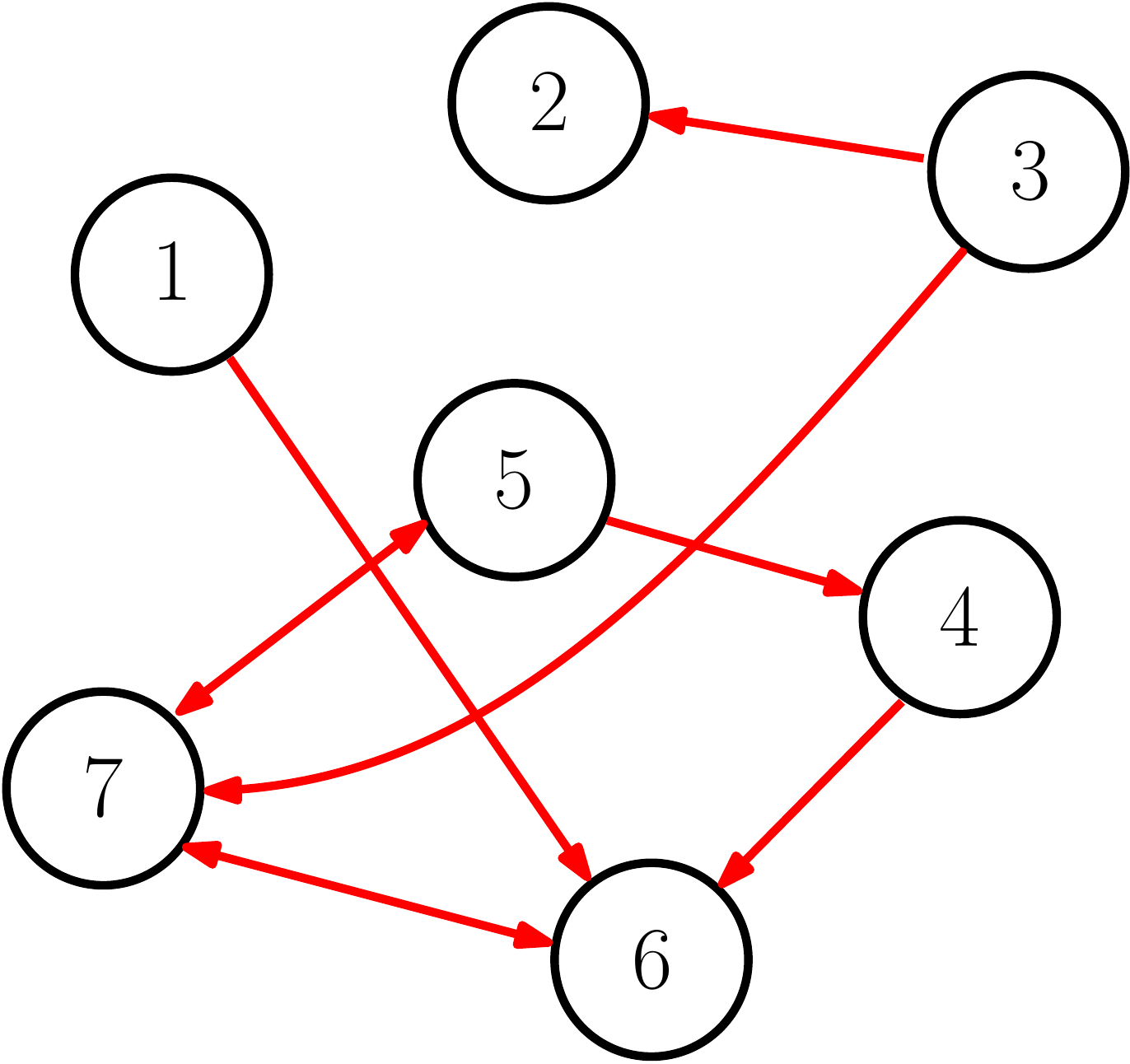}
\caption{Flow graph}
\label{fig:ex3_flow}
\end{subfigure}
\hfill
\begin{subfigure}[b]{0.49\columnwidth}
\centering
\includegraphics[width=1\textwidth]{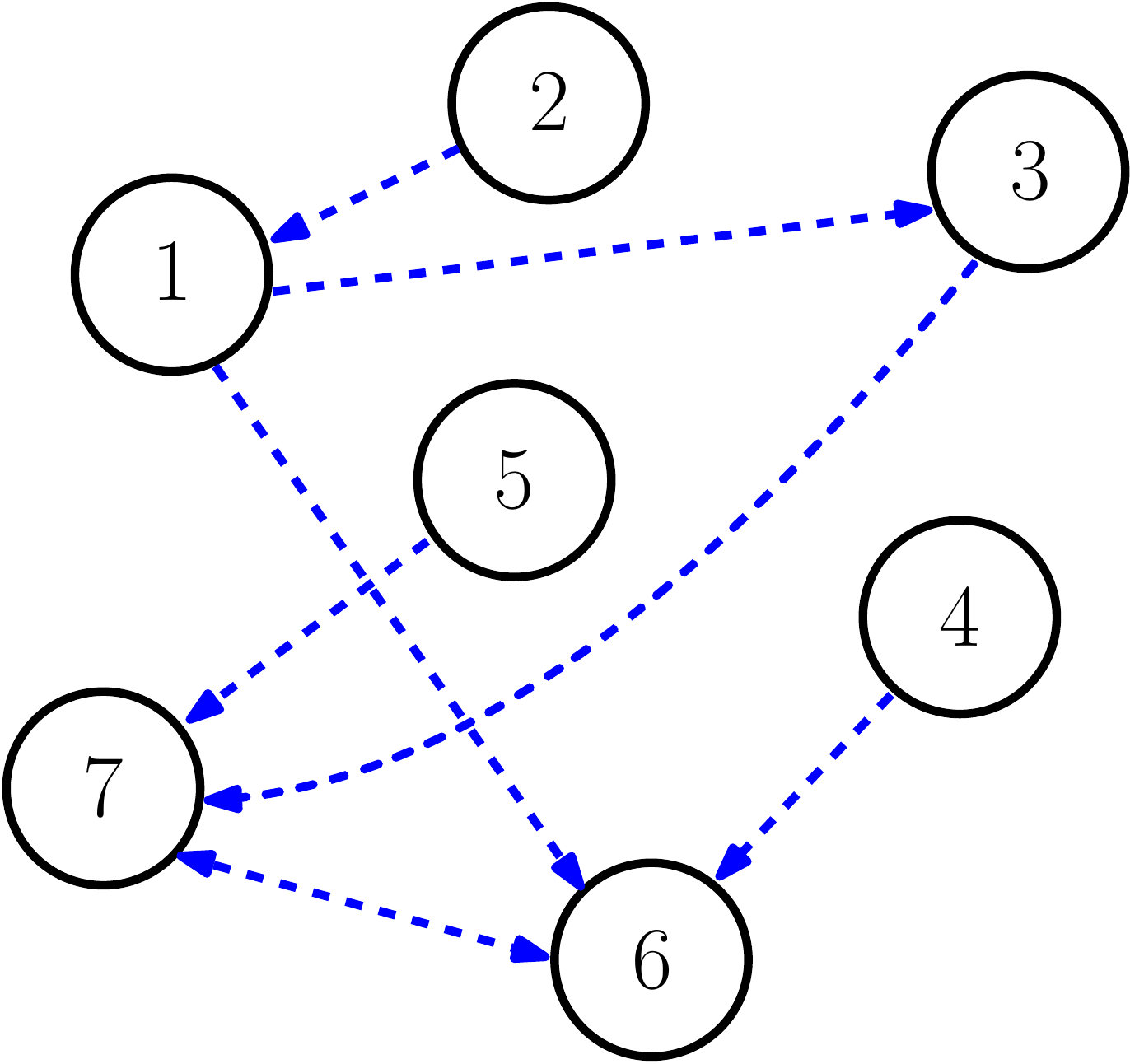}
\caption{Jump graph}
\label{fig:ex3_jump}
\end{subfigure}
\caption{Graphs for Example~3}\label{fig:ex3}
\end{figure}
Let us consider graphs as in Fig.~\ref{fig:ex3}. In this case, three consensuses are reached for the hybrid network according to the clusters defined by the AEP $\pi^\star_\mathcal{H} = \{\mathcal{H}_{\mathrm{un},1}, \mathcal{H}_{\mathrm{un},2}, \rho_3\}$ with $\mathcal{H}_{\mathrm{un},1}, = \{\nu_1, \nu_2, \nu_3\}$, $\mathcal{H}_{\mathrm{un},2}, = \{\nu_4, \nu_5\}$ and $\rho_3 = \{\nu_6, \nu_7 \}$. As in the previous cases, consensus over the reaches are constant and given by
\begin{align*}
     x_{1}^{ss} =& v_{\alpha,1}^\top \mathbf{x}_1(0, 0), \quad v_{\alpha,1}^\top = \frac{1}{11}\begin{pmatrix} 5 & 1 & 6 \end{pmatrix}\\
    x_{2}^{ss} =& v_{\alpha,2}^\top \mathbf{x}_2(0, 0), \quad v_{\alpha,2}^\top = \frac{1}{2}\begin{pmatrix} 1 & 1 \end{pmatrix}
\end{align*}
so getting $x_1^{ss} = 2$ and $x_2^{ss} = 4.5$. In this case, moreover, one gets that the last components of the eigenvectors in (\ref{eig:L_a}) are given by
\begin{align*}
    \gamma^1_{\alpha} = \gamma^2_{\alpha} = \frac{1}{2}\ones_2
\end{align*}
so getting that $\Gamma^\mu_\alpha = \text{span}\{\ones_2\}$ is both $M_\mathrm{f}$ and $M_\mathrm{j}$ invariant with
\begin{align*}
    M_\mathrm{f} = M_\mathrm{j} = \begin{pmatrix}
    3 & -1
    \\ -1 & 3
    \end{pmatrix}.
\end{align*}
Thus, by Corollary \ref{cor:1}, nodes in the cell $\rho_3$ associated with the common $\mathcal{C}_\mathrm{un}$ converge to a constant consensus which is given by
\begin{align*}
    x_{\delta, 1}(t, k) \to \frac{1}{2} (x_1^{ss}+ x_2^{ss})\ones_2
\end{align*}
so getting in this case the constant consensus $ x_{\delta, 1}(t, k) \to 3.25 \ones_2$ as depicted in Figure \ref{fig:sim_ex3}.
\begin{figure}[h!]
\centering
\includegraphics[width=\textwidth]{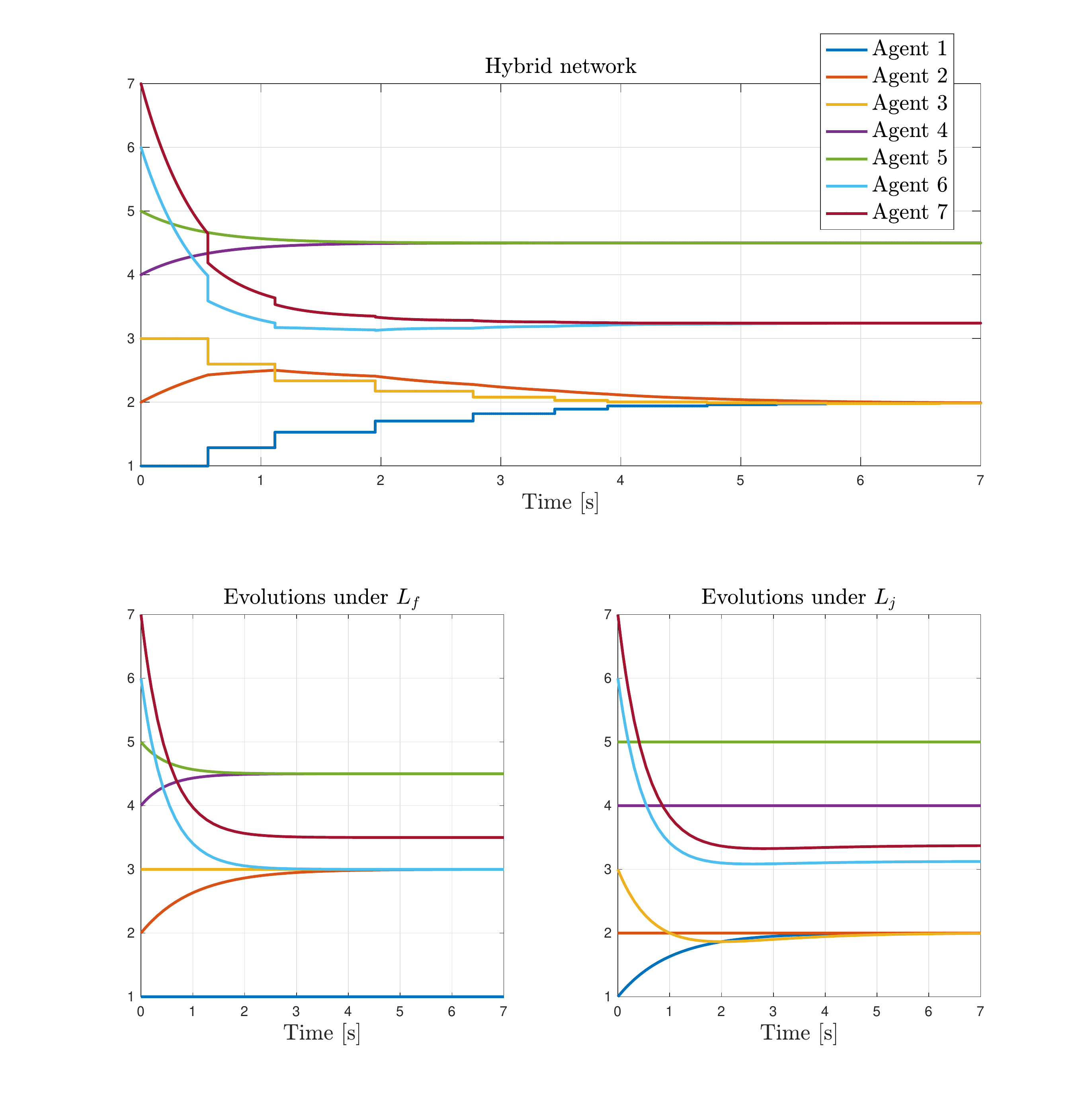}
\caption{Simulations for Example 3}
\hfill
\label{fig:sim_ex3}
\end{figure}
\end{example}

\section{Conclusions} \label{sec:conclusions}
In this paper, we have provided a full characterization of the multi-consensus for scalar hybrid systems under aperiodic time-driven jumps evolving over networks. Considering different communication graphs for the flow and jump dynamics, it has been shown that the clusters defining the behavior of the hybrid agents are identified by the coarsest partition which is almost equitable for both $\mathcal{G}_\mathrm{f}$ and $\mathcal{G}_\mathrm{j}$. More in detail, the number of multi-consensuses is strictly related to the union and intersection graphs associated with $\mathcal{G}_\mathrm{f}$ and $\mathcal{G}_\mathrm{j}$ with as many consensuses as the number of exclusive reaches of the union plus further clusters induced over the remaining common component by  the intersection graph. In addition, although consensuses over the reaches are constant behaviors, the asymptotic agreements over the common components are generally described by hybrid arcs. Future perspectives concern the characterization of multi-consensus for larger classes of hybrid systems with the inclusion of state-driven jumps and non-scalar agent dynamics with emphasis to the open problems detailed in \cite{MAGHENEM2020335}.

\section*{Acknowledgment}
The Authors wish to thank Prof. Salvatore Monaco for the fruitful discussions and remarkable remarks on these topics.

\bibliographystyle{elsarticle-harv} 
\bibliography{Biblio}
\end{document}